\documentclass[11pt, a4paper]{article}

\usepackage{amsmath, amsthm, amssymb, amsfonts}
\usepackage{lmodern} 
\usepackage[utf8]{inputenc}
\usepackage{mathtools}
\usepackage{hyperref}
\usepackage{graphicx}
\usepackage{dsfont}
\usepackage{fullpage}
\usepackage{color}
\usepackage{soul} 
\usepackage[title]{appendix}
\usepackage{tikz}
\usetikzlibrary{patterns}
\usetikzlibrary{shapes.multipart}
\usetikzlibrary{arrows}
\usepackage{braket}
\usepackage{stmaryrd}

\definecolor{labelkey}{cmyk}{.4,.2,0,0}

\newcommand{\be}{\begin{equation}}
\newcommand{\ee}{\end{equation}}
\newcommand{\bea}{\begin{eqnarray}}
\newcommand{\eea}{\end{eqnarray}}
\newcommand{\nn}{\nonumber}
\newcommand{\kk }{\mathsf{k}}

\newcommand{\EE}{\ensuremath{\mathbb{E}}}

\newcommand{\R}{\ensuremath{\mathbb{R}}}
\newcommand{\C}{\ensuremath{\mathbb{C}}}
\newcommand{\Z}{\ensuremath{\mathbb{Z}}}

\newcommand{\I}{\ensuremath{\mathbf{i}}}
\newcommand{\Real}{\ensuremath{\mathfrak{Re}}}

\newcommand{\Res}[1]{\underset{{#1}}{\mathbf{Res}}}

\renewcommand{\rho}{\varrho}
\newcommand{\e}{\epsilon}
\newcommand{\eps}{\varepsilon}
\renewcommand{\leq}{\leqslant}
\renewcommand{\geq}{\geqslant}

\newcommand{\strip}{\mathbb{S}}

\usepackage[backend=biber,style=numeric, maxnames=5, maxalphanames=5, sortcites=true, giveninits=true, isbn=false, clearlang=true, doi=false, url=false, sorting=none]{biblatex}
\renewbibmacro{in:}{}
\AtEveryBibitem{\clearlist{language}} 
\DeclareFieldFormat[article]{title}{\mkbibemph{#1}}
 \DeclareFieldFormat[article]{journaltitle}{#1}

\numberwithin{equation}{section}

\usepackage{titlesec}
\titleformat{\section}{\large\bf}{\thesection}{1em}{}
\titleformat{\subsection}[runin]{\bf}{\thesubsection}{1em}{}[.]
\titleformat{\subsubsection}[runin]{\it}{\thesubsubsection}{1em}{}[.]

\newtheorem{theorem}{Theorem}[section]

\newtheorem{lemma}[theorem]{Lemma}

\theoremstyle{definition}

\theoremstyle{definition}

\theoremstyle{definition}

\theoremstyle{definition}
\newtheorem{remark}{Remark}[section]

\usepackage{authblk}
\author[1]{Guillaume Barraquand and Pierre Le Doussal}
\affil[1]{\normalsize Laboratoire de Physique de l'\'Ecole Normale Sup\'erieure, ENS, Universit\'e PSL, CNRS, Sorbonne Universit\'e, Universit\'e Paris-Cité, 75005 Paris, France}

\title{\bf \large 
Large time cumulants of the KPZ equation on an interval} 
\date{}
\begin{document}

\maketitle

\begin{abstract} 
We consider the Kardar-Parisi-Zhang equation on the interval $[0,L]$ with Neumann type boundary conditions
and boundary parameters $u,v$. 
We show that the $k$-th order cumulant of the height behaves as $c_k(L,u,v)\,  t$ in the large time limit $t \to +\infty$,
and we compute the coefficients $c_k(L,u,v)$. We obtain an expression for 
the upper tail large deviation function of the height. We also consider the
limit of large $L$, with $u=\tilde u/\sqrt{L}$, $u=\tilde v/\sqrt{L}$,
which should give the same quantities for the two parameter family $(\tilde u, \tilde v)$ KPZ fixed point on the interval. 
We employ two complementary methods. On the one hand we adapt to the interval the replica Bethe ansatz method pioneered by Brunet and Derrida
for the periodic case. On the other hand, we perform a scaling limit using previous results available for the open ASEP.
The latter method allows to  express the cumulants of the KPZ equation in terms
a functional equation involving an integral operator. 
\end{abstract}

\section{Introduction and main results} 
\subsection{Overview}

The probability distribution of the solution for the height field $h(x,t)$ of the Kardar-Parisi-Zhang equation \cite{kpz} on $\R$ is known exactly, at least for the narrow wedge initial condition \cite{calabrese2010free, dotsenko2010replica, sasamoto2010exact, amir2011probability}. On $\mathbb R_+$ with Neumann type boundary condition, the probability distribution is known as well, at least near the boundary \cite{gueudre2012directed, barraquand2018stochastic, krajenbrink2020replica, imamura2022solvable}. In this article, we consider the case of the KPZ equation in a finite volume in one dimension, which is richer, and more challenging: the probability distribution of the solution is not known exactly. There are two types of boundary conditions that one may impose for the KPZ equation on $[0,L]$: the periodic case (i.e. the KPZ equation on a ring), and the case of Neumann type boundary conditions (i.e. the KPZ equation on an interval, also called the open KPZ equation).
\medskip 

The case where $t$ and $L^{3/2}$ are of the same order and large is particularly interesting, and related to the KPZ fixed point \cite{matetski2021kpz}. For periodic boundary conditions, the periodic KPZ fixed point was studied in \cite{prolhac2016finite, baik2018fluctuations, baik2019multi,baik2021periodic,prolhac2024kpz} through a limit of the periodic TASEP. For the KPZ equation on an interval, the exact distribution of the open KPZ fixed point is unknown. 
\medskip

In the regime where $L \gg t^{2/3}$, the problem reduces to the case of the KPZ equation on $\R$ or $\R_+$, depending on whether one is interested in the distribution of $h(t,x)$ for $x$ in the bulk or at the edge of the interval. In this article, we consider the opposite limit, i.e. when $t \gg L^{3/2}$, for fixed $L$. 
\medskip

For fixed $L$ and $t$ going to infinity, the height $h(x,t)$ grows as $t$ goes to infinity as $c_1 t$, and $\frac{h(x,t)-c_1 t}{\sqrt t}$ converges to a Gaussian random variable with variance $c_2$. 
In the periodic case, the constants $c_1$ and $c_2$,  as well as the higher cumulants in the large time limit, were computed by Brunet and Derrida in \cite{brunet2000probability,brunet2000ground} via the replica method. This allowed them to compute the upper tail large deviation function of the height field $h(x,t)$ for large $L$ (i.e. for $t\to\infty$ and then $L\to\infty$).  In this paper, we compute these cumulants in the case of the KPZ equation in an interval,
with general boundary parameters. The cumulant generating function is obtained in a parametric form, depending on a meromorphic function which solves a functional equation involving an integral operator, denoted $\kk$ below -- see Equation  \eqref{eq:functionalequation}. 
\medskip

We employ two completely different methods. On the one hand, we show that the original method of Brunet and Derrida can be adapted to the case of the interval. We will proceed slightly differently from \cite{brunet2000probability,brunet2000ground} and we obtain formulas involving another integral operator, denoted $K$ below. 
On the other hand we will obtain the cumulants for the KPZ equation on an interval through
a limit of the open asymmetric simple exclusion process (ASEP). Indeed, for the totally asymmetric simple exclusion process (TASEP), Prolhac  \cite{prolhac2010tree} showed that the cumulant generating function can be defined as a parametric equation involving the solution to some functional equation. The same functional equation was recovered using a deformed version of the matrix product ansatz in \cite{lazarescu2011exact, LazarescuThesis} and this method was adapted to the case of the interval with open boundary conditions and to the ASEP \cite{gorissen2012exact, lazarescu2013matrix} -- see also \cite{lazarescu2014bethe} for another derivation, and  \cite{mallick2015exclusion} for a review. 
We show that if we scale all parameters so that the open ASEP converges to the KPZ equation on an interval, the functional equation admits a particularly nice limit which allows for the iterative computation of cumulants. However, there is a subtle issue of exchange of limits here: 
the convergence of the ASEP to the KPZ equation is about the typical behaviour of the height function. A priori, the convergence in distribution of ASEP to KPZ does not ensure that the large deviations of both models are related in the limit. Nevertheless, we check for the first two cumulants that the large time cumulants of the KPZ equation on an interval are limits of those of the ASEP, by matching the results with what we obtain using the method of Brunet and Derrida. The formulas obtained by both methods are different, but equivalent.  In particular, the integral operators ($\kk$ in Section \ref{sec:mainresults}, Section  \ref{sec:limitLazarescu} and \ref{sec:cumulants} and $K$ in Section \ref{sec:BrunetDerrida}) are slightly different. 
\medskip

In addition to these results at fixed $L$, we also consider the large $L$ limit of the cumulants, when boundary parameters are scaled as  $u=\tilde u/\sqrt{L}$, $u=\tilde v/\sqrt{L}$. We expect that this limit yields the large time cumulants of the two parameter $(\tilde u, \tilde v)$ KPZ fixed point on the interval. These cumulants can also be computed as a large scale large time cumulants of TASEP -- see \cite{lazarescu2011exact} when $\tilde u,\tilde v$ are $\pm \infty$ or $0$ and \cite{prolhac2024kpz, godreau2021riemann} for the general case.

\medskip 

Finally, our calculation provides an alternative route to solve again the periodic case, leading to results in a different form. For the first two cumulants it is easy to show equivalence  with the results of \cite{brunet2000probability,brunet2000ground}. Another representation of
these first two cumulants was obtained recently in terms of a functional of Brownian bridges 
\cite{gu2024kpz, dunlap2023fluctuation, brunet2021high, gu2025some}. It
would be quite interesting to extend these representations to the interval,
and to relate them to the explicit analytical expressions both for the periodic
case and for the interval (see Appendix \ref{sec:appendixfunctionals}).

\subsection{Model and observables}

We consider the KPZ equation for the evolution in time $t>0$ of the height field $h(x,t)$ on the interval $x \in [0,L]$ 
\be 
\partial_t h(x,t) = \frac{1}{2} \partial^2_{x} h(x,t) + \frac{1}{2} \left(\partial_x h(x,t) \right)^2 + \xi(x,t)
\label{eq:KPZ}
\ee 
where $\xi(x,t)$ is a standard space-time white noise, 
with boundary conditions parametrized by $(u,v) \in \mathbb{R}^2$
\be
\partial_x h(t,x) \Big\vert_{x=0} =u, \quad  \quad \partial_x h(t,x) \Big\vert_{x=L} =-v.
\label{eq:boundary}
\ee 
This stochastic PDE is often called the open KPZ equation, as it arises as a limit of open ASEP,
by contrast with the periodic KPZ equation which arises as a limit of the ASEP on a ring.
The solution can be written as $h(x,t)=\log Z(x,t)$, 
in terms of the partition sum $Z(x,t)$ of a directed polymer, which satisfies the stochastic heat equation (SHE)
\be 
\partial_t Z(x,t) = \frac{1}{2} \partial^2_{x} Z(x,t) + Z(x,t)\xi(x,t)
\label{eq:SHE}
\ee 
with boundary conditions 
\be  \label{eq:SHEboundaryconditions} 
\partial_x Z(x,t)\Big\vert_{x=0} = (u-\tfrac{1}{2})Z(0,t), \;\;\;\; \partial_x Z(x,t)\Big\vert_{x=L} = -(v-\tfrac{1}{2})Z(L,t).
\ee 
\subsubsection{Stationary measures} 
\label{sec:stationarymeasures}
Before describing the observables of interest here, it is useful to recall the stationary measure of the open KPZ equation, which was obtained recently   
\cite{corwin2024stationary, barraquand2022steady, bryc2023markov, barraquand2023stationary,
barraquand2024stationary, corwin2022some}. 
We say that a process $x\mapsto H(x)$ is stationary for the open KPZ equation if $h(0,x)-h(0,0)=H(x)$ implies that  for all $t>0$, $h(t,x)-h(t,0)$ has the same distribution as $H(x)$ (as a process in the spatial variable $x$). The distribution of $H(x)$ may also be viewed as  the large time limit of the distribution of $h(t,x)-h(t,0)$. 

\medskip 

The stationary height field is given by $H(x)=\Lambda_1(x)-\Lambda_1(0)$ where the processes $(\Lambda_1,\Lambda_2)$ are a couple of  Brownian motions interacting through an exponential potential.
They are distributed, for $u+v>0$, as \cite{bryc2023markov, barraquand2022steady}
\begin{equation}
\frac{1}{\mathcal Z_{u,v}(L)} e^{-u(\Lambda_1(0)-\Lambda_2(0)) -v(\Lambda_1(L)-\Lambda_2(L)) -\int_{0}^L 
\left[ \frac{1}{2}\left( \frac{d\Lambda_1(x)}{dx}\right)^2 + \frac{1}{2} \left( \frac{d\Lambda_2(x)}{dx}\right)^2   + e^{-(\Lambda_1(x)-\Lambda_2(x))}\right]dx }
\label{eq:KPZstat}
\end{equation}
where $\mathcal Z_{u,v}(L)$ is a normalization constant given below\footnote{Stritly speaking, \eqref{eq:KPZstat} defines an infinite measure, but if we fix $\Lambda_1$ or $\Lambda_2$ at one point, say $\Lambda_2(0)=0$, then \eqref{eq:KPZstat} becomes a finite measure, and $\mathcal Z_{u,v}(L)$ is precisely chosen so that it is a probability measure. By translation invariance, one can fix either $\Lambda_1$ or $\Lambda_2$ at any point $x_0$, this would not change the probability measure of increments for both processes.}. 
Alternatively, the stationary height field may be defined as $H(x)=W(x)+ X(x)$ where $W$ is a standard Brownien motion with diffusion coefficient $1/2$ and $W(0)=0$, and $X$ is distributed as \cite{barraquand2022steady}
\be 
\frac{\delta_0(X(0))}{\mathcal K_{u,v}(L)} \left( \int_0^L e^{-2X(x)}dx \right)^{-u-v}e^{-2vX(L)} e^{-\int_0^L \left(\frac{dX(x)}{dx} \right)^2 dx}.
\label{eq:stationarymeasure2}
 \ee
 with $X(0)=0$ and $\mathcal K_{u,v}(L) = \mathcal Z_{u,v}(L)/\Gamma(u+v)$. 
 This last expression holds for any $u,v\in \R$. 

\subsubsection{Observables of interest}
In this paper, we are interested in the large time asymptotics of the cumulants of the height
\be 
c_k(u,v,L)  := \lim_{t\to\infty} \frac{1}{t} \kappa_k(h(x,t)), \label{defck} 
\ee 
where, for a random variable $X$, $\kappa_k(X)$ denotes the $k$-th cumulant of $X$. The cumulants  $\kappa_k(h(x,t))$ depend on $x$ and the initial condition $h_0(x)=h(x,0)$. When $h_0$ is chosen according to the stationary measure on the interval, as described above, the first cumulant is exactly linear in $t$, i.e. $\kappa_1(h(0,t)) = t c_1(u,v,L)$. For an arbitrary initial condition $h_0$, since the  solution converges to the unique stationary state, the result of the limit $\lim_{t\to\infty} \frac{1}{t} \kappa_1(h(0,t))$ must be the same as starting from the stationary state.

More generally, since we work on an interval of finite size $L$, we expect that the value of $c_k(u,v,L)$ neither depend on $x$.
It is interesting to note that there are subleading corrections to \eqref{defck} at large time, which depend on both $x$ and $h_0$ \cite{prolhac2024approach},
that we do not address here. We can thus from now on choose $x=0$. 

\subsubsection{Cumulant generating functions and upper tail large deviations}
A related observable is the cumulant generating function 
\be 
E_t(s) = \frac{1}{t} \log \mathbb E\left[ e^{s h(0,t)}\right] = \frac{1}{t} \sum_{k=1}^{\infty} \frac{s^k}{k!} \kappa_k(h(0,t)).  
\ee 
As $t\to\infty$, $E_t(s)$ converges to a function $E(s)$, which depends on $u,v,L$. 
We expect that the convergence is such that one can recover the large time cumulants $c_k(u,v,L)$ from $E(s)$, that is we have 
\be \label{Emuexpansion}
E(s) = \sum_{k=1}^{\infty} \frac{s^k}{k!} c_k(u,v,L)
\ee

Finally, denoting by $P(h)$ the PDF of $h(0,t)$, we expect the upper tail large deviations at large time 
\be  \label{defPhi} 
P(h) \sim e^{- t \,  \Phi\left(\frac{h}{t} +\frac{1}{24} \right)}, 
\ee 
where $\Phi(y)$ is a rate function (here $\sim$ means that the logarithm of both sides are equivalent as $t$ goes to infinity). More precisely, as we expect the same tail for the PDF and the CDF, we expect that for all $y>c_1+\frac{1}{24}$, 
\be 
\lim_{t\to\infty} \frac{\log \mathbb P\left( h(0,t)+\frac{t}{24} >y t \right)}{t} = \Phi(y).
\label{eq:defPhimaths}
\ee 
Letting $H=\frac{h}{t}+\frac{1}{24}$, a saddle-point method  leads to 
\be  \label{legendre1} 
 E( s) =  \frac{-s}{24} +  \max_{H \in \R} \left\lbrace  s H - \Phi(H) \right\rbrace.
\ee 
This can be inverted to obtain the rate function from $E(s)$ as 
\be  \label{legendre2}
\Phi(H) =  \max_{s \in \R}  \left\lbrace  sH   -  \left(E(s)+\frac{s}{24}\right) \right\rbrace.
\ee 
In this paper, we will compute $E(s)$ in an interval around $s=0$, hence characterizing the cumulants $c_k(u,v,L)$.

\subsection{Main results} 
\label{sec:mainresults}

\subsubsection{Functional equation} \label{subsubsec:functional} 
Assume that $u,v>0$. 
The (asymptotic) cumulant generating function $E(s)$
is obtained parametrically through the system of equations 
\begin{align} \label{systemEmu1} 
    s &= \int_{\I\R} \frac{dw}{2\I\pi} U(w),\\
    E(s) &= - \frac{1}{24} s + \frac{1}{2} \int_{\I\R} \frac{dw}{2\I\pi} w^2U(w),  \label{systemEmu2}
\end{align}
where the  function $U(w)$ is the solution to the functional equation 
\be 
U(w) = \frac{-1}{2} \log\left( 1-2 \zeta \Psi(w) e^{\kk U(w)} \right),
\label{eq:functionalequation}
\ee 
and $\zeta$ is a free variable that must be eliminated between the two above equations \eqref{systemEmu1}-\eqref{systemEmu2}. Here, in the case of the interval, one has 
\be \label{defPsi}
\Psi(w) = \frac{\Gamma( u + w) \Gamma(u-w) \Gamma(v + w) \Gamma(v-w) e^{w^2L}}{\Gamma( 2w) \Gamma(-2 w)}
\ee 
and the linear operator $\kk$ is defined by 
\be \label{defoperator} 
\kk \phi(w) = \int_{\I\R} \frac{dw'}{2\I\pi} \kk(w,w')\phi(w')
\ee
where
\be \label{defkernel} 
\kk(w,w') = -2\left(\psi(1+w-w') + \psi(1+w'-w )\right)\,\, \text{ with } \psi(z) = \partial_z\log \Gamma(z).
\ee
Assuming that $U(w) = \sum_{k=1}^{\infty} \zeta^k U_k(w)$ and that both sides of the functional equation \eqref{eq:functionalequation} can be expanded in powers series of $\zeta$, one can show that the solution is unique. The digamma function $\psi(z)$ is a meromorphic function with simple poles at $z = -n$, $n \in \mathbb{N}$ and satisfies the following useful properties 
\be \label{properties}
\psi(z+1)=\psi(z)+ \frac{1}{z},  \quad  \quad \int_0^{+\infty} dt \frac{e^{- t z}-1}{e^t -1} = - \gamma_E - \psi(1+z).
\ee

\begin{remark}
The function $\Psi$ appears in \cite{corwin2024stationary, bryc2023markov, barraquand2022steady, barraquand2024integral}. For $u,v>0$, the normalization constant in \eqref{eq:KPZstat} is given by \cite{bryc2023markov, barraquand2022steady}
\be \label{norm} 
\mathcal Z_{u,v}(L) = \frac{1}{2} \int_{\I\R} \frac{dw}{2\I\pi} \Psi(w).
\ee 
\end{remark}
\begin{remark}
The formulas above are stated for $u,v>0$. As we will see, results when $u\leq 0$ or $v\leq 0$ can be obtained after performing analytic continuations in the parameters $u,v$. 
\end{remark}

\subsubsection{First and second cumulants}

From the above functional equation one obtains the cumulants. Let us give the explicit expressions
of the first few cumulants. We have  that, for $u,v>0$, the first cumulant reads \cite{barraquand2024integral}
\be \label{eq:firstorderintro}
c_1(u,v,L) =\frac{-1}{24}+ \frac{1}{2}\frac{\int_{\I\R} \frac{dw}{2\I\pi} w^2 \Psi(w)}{\int_{\I\R} \frac{dw}{2\I\pi}  \Psi(w)}
= \frac{-1}{24}+ \frac{1}{2} \partial_L \log \mathcal Z_{u,v}(L) 
\ee 
where the second term can be written as a derivative w.r.t. $L$ of the logarithm of the normalization
in \eqref{norm}, using that 
\be \label{PsiderivativeL}
\partial_L \Psi(w) = w^2 \Psi(w).
\ee 
Let us now define the normalized measure $\nu(dw)$
\be 
\nu(dw) = \frac{\Psi(w)}{  \int_{\I\R} \frac{dw'}{2\I\pi} \Psi(w') } \frac{dw}{2\I\pi}.
\label{eq:defnu}
\ee 
The second cumulant is given, for $u,v>0$, by 
\be \label{c2exact} 
c_2(u,v,L) =  
\int_{\I\R} \nu(dw_1) \int_{\I\R} \nu(dw_2) ( \kk(w_1,w_2) +  \delta(w_1-w_2) ) (w_1^2 - \int_{\I\R } \nu(dw) w^2 ),
\ee 
where the delta function above is defined by 
\be \label{defdelta} 
\int_{\I\R} \frac{dw_2}{2\I\pi} \delta(w_1-w_2) f(w_2) = f(w_1).
\ee 
Using the identity
\be \label{c2derivativeL}
 \partial_L \nu(dw) = \left( w^2 - \int_{\I\R } \nu(dw') (w')^2 \right) \,  \nu(dw) 
\ee 
which is a simple consequence of \eqref{PsiderivativeL} and \eqref{eq:defnu},
we see that the second cumulant can also be written as a total derivative (with respect to  $L$) of a simpler expression: 
\be \label{eq:secondcumintro}
c_2(u,v,L) =  \frac{1}{2} \partial_L \left( \int_{\I\R} \nu(dw_1) \int_{\I\R} \nu(dw_2) \bar \kk(w_1,w_2) \right)
\ee 
where we have defined the reduced kernel 
\be  \label{defkbar} 
 \bar \kk(w_1,w_2) = \kk(w_1,w_2) + \delta(w_1-w_2).
\ee 
Figure \ref{fig:plot} is a plot of $c_2(u,u,L=1)$ which shows that the second cumulant is not monotone as $u$ varies in $\R_+$. As $u \to +\infty$, it converges to a finite value, associated to Dirichlet boundary conditions $Z(0,t)=Z(L,t)=0$. 
\begin{figure}
    \centering
    \includegraphics[width=0.6\linewidth]{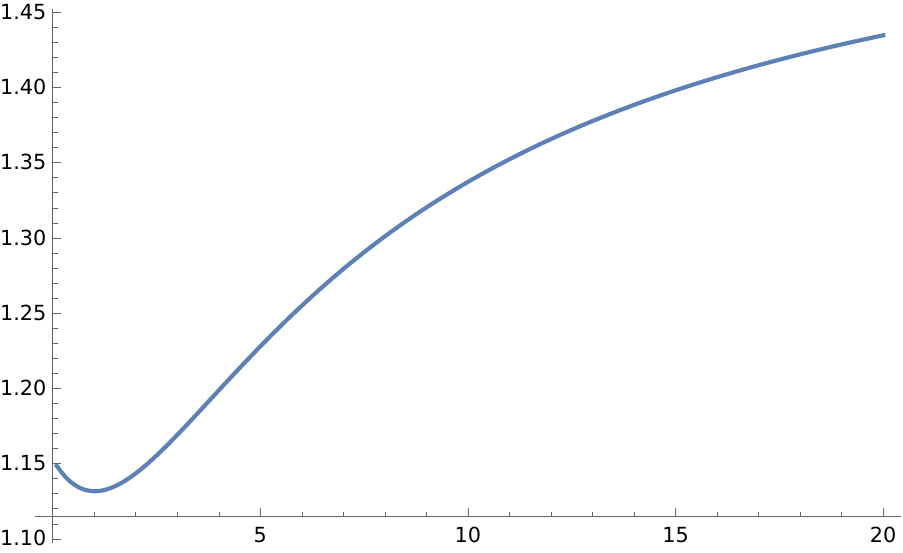}
    \caption{Graph of the function $u\mapsto c_2(u,u,1)$.}
    \label{fig:plot}
\end{figure}
\medskip 

Since $L$ is finite, we expect that the cumulants are analytic functions  in $u$ and $v$. Thus,  for $u,v$ not necessarily positive, the cumulants are given by analytic continuations (in $u$ or $v$ or both variables) of the expressions above, see Section \ref{sec:analytic}. In particular, when $u=-v$, 
\be 
c_1(u,-u,L)=\frac{-1}{24}+  \frac{1}{2}u^2
\ee 
which is the usual velocity of the KPZ height function on $\R$ for Brownian initial data with drift $u$. 
Regarding the second cumulant, it is given for $u+v=0$ by 
\be
c_2(u,-u,L) = \vert u\vert-\frac{1}{2 e^{Lu^2}} \int_{\I\R} \frac{dw}{2\I\pi} \Psi(w)\Big\vert_{v=-u}(w^2-u^2). 
\label{eq:c2u+v=0}
\ee 
In particular, when $u=v=0$, 
\be
c_2(0,0,L)=  \int_{\R}  x \coth(\pi x) e^{- L x^2} dx.
\label{eq:c200intro}
\ee

\begin{remark}
    The action of the operator $\kk$ on functions $\phi$ can be interpreted concretely. If we can write $\phi$ as 
    $$\phi(w) =f(w)-f(w+1)+ g(-w)-g(1-w)$$
    where the functions $f$ and $g$ are meromorphic functions with no poles to the right of $\I\R$, 
    \be 
\kk \phi(w) = 2 f(1+w) + 2g(1-w), 
\ee 
    see Appendix \ref{sec:operatork} for details.  The functions $f$ and $g$ also have an interpretation in terms of Bethe roots, see  \cite[Appendix E]{lazarescu2014bethe}. 
\end{remark}
\begin{remark}
    We expect that the first and second cumulant can be expressed as functionals of the stationary measure, as it is the case for periodic boundary conditions \cite{dunlap2023fluctuation, gu2024kpz}. For the first cumulant, we provide such expressions as functionals of Brownian motions in some special cases in Appendix \ref{sec:appendixfunctionals}. 
\end{remark}

\subsubsection{General expression for the cumulants}

To obtain the general cumulants we expand the equations \eqref{systemEmu1}, \eqref{systemEmu2}
and\eqref{eq:functionalequation} in powers of $\zeta$ and eliminate $\zeta$ order by order.
As a first step we write $U(w) = \sum_{n \geq 1}  U_n(w) \zeta^n$ and determine the coefficients $U_n(w)$. From 
\eqref{eq:functionalequation} we find
that they can
be computed by recurrence as 
\be 
U_n = \frac{2^{n-1}}{n} \Psi^n + \sum_{\ell =1}^{n-1}   \frac{2^{\ell-1} }{\ell} \Psi^\ell  \sum_{i=1}^{n-\ell} \frac{\ell^i}{i!}  \sum_{\begin{matrix} (p_1, \dots, p_i) \in \mathbb Z_{\geq 1}^i \\ p_1+\dots+p_i=n-\ell \end{matrix} } \prod_{j=1}^i \kk  U_{p_j}(w)
\label{eq:recurrenceforU}
\ee 
with $U_1(w)= \Psi(w)$. 

We are interested in coefficients of $E(s)$ in its expansion of powers of $s$, but equations \eqref{systemEmu1} and \eqref{systemEmu2} give us access to the expansion of $s$ and $E(s)$ in powers of $\zeta$. In other terms, we have access to coefficients $s_i$ and $\eps_i$ such that 
\be 
s = \sum_{n \geq 1}  s_n \zeta^n, \,\, E(s) = \sum_{n \geq 1} \eps_n \zeta^n.
\label{eq:relationcoeffs}
\ee
where 
\be 
s_n = \int_{\I\R} \frac{dw}{2\I\pi} U_n(w), \,\,\,\eps_n = \int_{\I\R} \frac{dw}{2\I\pi} w^2 U_n(w).
\ee 
Then,  we want to find the cumulants $c_k=c_k(u,v,L)$ such that $ E(s)=\sum_{k \geq 1}  \frac{s^k}{k!} c_k$, see \eqref{Emuexpansion}. Let 
\be d_{k,\ell}= \sum_{\begin{matrix}(i_1, \dots, i_{\ell})\in \mathbb Z_{\geq 1}^{\ell} \\ i_1+\dots+i_{\ell}=k\end{matrix} } \prod_{j=1}^{\ell} s_{i_j}.\ee
The cumulants $c_k$ can be computed through the recurrence 
\be  \label{recurrenceck}
\frac{c_k}{k!} = \frac{1}{s_1^k} \left(\eps_k - \sum_{\ell=1}^{k-1} d_{k,\ell} \frac{c_\ell}{\ell!} \right).  
\ee
This procedure allows to compute the $c_k$ iteratively and recovers the expressions for $c_1$ and $c_2$ given above.
The expression of $c_3$ is given below.

\subsubsection{Universal limit as $L\to\infty$, KPZ fixed point on the interval}
We  scale boundary parameters as $u=\frac{\widetilde u}{\sqrt{L}}$ and $v=\frac{\widetilde v}{\sqrt{L}}$,
with $\tilde u>0, \tilde v>0$. 
This is the natural scaling introduced in \cite{barraquand2022steady} which 
leads to a 2 parameter family of processes which were found
to define the stationary measure of the KPZ fixed point on the interval. 
As $L$ goes to infinity, we find that the cumulants behave as 
\begin{align}  \label{cklargetime} 
c_1(u,v,L) &= \frac{-1}{24} + \frac{1}{L} \widetilde c_1(\widetilde u, \widetilde v)+o\left(\frac{1}{L}\right)\\   
c_k(u,v,L) &= L^{\frac{k-3}{2}} \widetilde c_k(\widetilde u, \widetilde v) + o\left( L^{\frac{k-3}{2}} \right) \text{ for }k\geq 2. 
\label{ckklargetime} 
\end{align}
The $\widetilde c_k$ can be computed by extracting coefficients of $\widetilde E(\widetilde s) = \sum_{i=1}^{+\infty} \frac{\widetilde s^i}{i!} \widetilde c_k$ where 
\begin{align}
\widetilde s &=\frac{1}{2} \sum_{k \geq 1} \frac{1}{k} \widetilde \zeta^k \phi_k, \quad \phi_k = \int_{ \mathbb{R}} \frac{dy}{ 2 \pi}  \phi(y)^k  \\ 
\widetilde E &=- \frac{1}{4} \sum_{k \geq 1} \frac{1}{k} \widetilde \zeta^k \psi_k, \quad \psi_k  = \int_{ \mathbb{R}} \frac{dy}{ 2 \pi} y^2  \phi(y)^k .
\end{align}
where 
\be 
\phi(y) =  \frac{4 y^2}{(\widetilde u^2 + y^2)(\widetilde v^2 + y^2)} e^{-y^2}.  
\ee 
The $c_k$ can then be obtained from a recurrence similar to 
\eqref{recurrenceck}, see Section \ref{sec:cumulantslargeL} where formula  
are given for the lowest cumulants. 

\medskip 

The function $\widetilde E(\widetilde s)$ is a limit of $E(s)$ in the sense that 
\be 
\widetilde E(\widetilde s) = \lim_{L\to\infty} L^{3/2} \left( E(\widetilde s/\sqrt{L}) +\frac{\widetilde s}{24\sqrt{L}} \right)
\ee 
Upon the Legendre transformation in \eqref{legendre2} 
we obtain that the upper tail large deviation rate function of the height, defined in \eqref{eq:defPhimaths} 
takes at large $L$ the scaling form (making explicit the dependence in $L$)
\be 
\Phi_L(\widetilde H/L) \simeq  L^{-3/2} \varphi( \widetilde  H)  
\ee 
where the scaling function $\varphi(\widetilde H)$ is independent of $L$. As we discuss below in Remark \ref{rem:universalityGuillaume}, $\varphi$ is the large deviation rate function of the two-parameter KPZ fixed point on the interval.  It is determined from 
the Legendre transform in \eqref{legendre2}, which leads to the parametric system
\bea
&& \varphi'(\widetilde H) = \widetilde  s  \\
&& \widetilde  H  = \widetilde E'(\widetilde s)  
\eea 
which together with $\varphi(0)=0$ allows to compute $\varphi(\widetilde H)$.

\begin{remark}
The scaling \eqref{cklargetime}, \eqref{ckklargetime} agrees with the KPZ scaling in the following sense. One 
expects that for both large $t,L$ with $L \sim t^{2/3}$ the cumulants obey scaling forms of the type
\be 
\kappa_k(h(0,t))  = t^{k/3} f_k(L/t^{2/3}).
\ee 
If we assume
that this scaling form matches with 
$\kappa_k(h(0,t)) \sim \widetilde c_k  t L^{\frac{k-3}{2}}$,
in the limit studied here, i.e. $t \gg L^{3/2} \gg 1$,
we conclude that $f_k(y) \sim_{y \to 0}  \widetilde c_k y^{\frac{k-3}{2}}$. If that 
is the case, although the functions $f_k$ may depend on the initial condition, their behavior at small argument
should be universal. 
\label{rem:universalityPierre}
\end{remark} 

\begin{remark}
    The scaling above is consistent with the prediction that $h(x,t)$ converges to the KPZ fixed point, i.e. 
 \be  
 \frac{1}{\sqrt{L}} \left( h(L y, L^{3/2}\tau)+\frac{L^{3/2}\tau}{24}\right) \xRightarrow[L\to\infty]{}\mathfrak h(y,\tau). 
 \label{eq:cvKPZfixedpoint}
 \ee
   Indeed, let us assume that \eqref{eq:defPhimaths} holds as well when $L$ goes to infinity, as long as $t\gg L^{3/2}$. This suggests to let $t=L^{3/2}\tau$, and send simultaneously $\tau$ and $L$ to infinity in \eqref{eq:defPhimaths}. Assuming that the limits in $\tau$ and $L$  can be taken in any order, we find using \eqref{eq:cvKPZfixedpoint}
   
   \begin{align} 
   \lim_{L\to\infty} L^{3/2} \Phi_L(v/L)&= \lim_{L\to\infty}  L^{3/2}
\lim_{\tau\to\infty}  \frac{-1}{\tau L^{3/2}} \log \mathbb P\left( \frac{ h(0,\tau L^{3/2})}{\tau L^{3/2}} >\frac{-1}{24}+\frac{v}{L} \right) \\&= 
\lim_{\tau\to\infty} \frac{-1}{\tau } 
\lim_{L\to\infty}   \log \mathbb P\left( \frac{ h(0,\tau L^{3/2})}{\tau L^{3/2}} >\frac{-1}{24}+\frac{v}{L} \right) \\
&= \lim_{\tau\to\infty}   \frac{-1}{\tau } \log \mathbb P\left( \frac{ \mathfrak h(0,\tau)}{\tau } > v \right)
 \end{align} 
 Hence, under this assumption of exchange of limits, we recover the fact that 
 \be 
\Phi_L(v/L) \sim \frac{1}{L^{3/2}} \varphi( v)
 \ee 
where $\varphi$ is the upper tail large deviation rate function for the KPZ fixed point, i.e. 
\be 
\varphi(y) = \lim_{\tau\to\infty} \frac{-1}{\tau} \log \mathbb P\left( \frac{ \mathfrak h(0,\tau)}{\tau } > y \right).
\ee 

\label{rem:universalityGuillaume}
\end{remark}

\subsection{Mathematical aspects}
The results stated in the present paper are obtained using methods introduced in the physics literature and they are not proved according to the standards of rigor of the mathematics literature. Let us comment on the main mathematical issues raised by these methods,  and potential ways to circumvent them. As explained above, we use two different methods. The first one,  presented in Section \ref{sec:limitLazarescu},  relies on taking a limit of the large time cumulants of the ASEP. As we have already mentioned, this poses the following issue:  it is proved that the fluctuations of ASEP height function converges to the KPZ equation \cite{corwin2018open, parekh2019kpz},  but there is no guarantee that the large time limit of the cumulants -- which encode the upper tail large deviation rate function -- should converge as well. At the level of ASEP, the expression for the (large time) cumulant generating function relies on the computation of the top eigenvalue of a deformed Markov matrix. This is a standard trick (see e.g. \cite{derrida1998exact, lebowitz1999galavotti}) and we believe that it can be mathematically justified. The computation of this top eigenvalue relies on a perturbative argument in the original references  \cite{gorissen2012exact, lazarescu2013matrix}, but an exact derivation is proposed using the functional Bethe ansatz in \cite{lazarescu2014bethe}. It could be interesting to adapt the latter method to other integrable probabilistic  models such as the stochastic six-vertex model. 

\medskip 
The most obvious way to avoid the issue described above is to work directly at the level of the KPZ equation. This is what we do in Section \ref{sec:BrunetDerrida}. However, the approach we employ there suffers several mathematical issues. First, it relies on the replica trick, meaning that we interpolate an unknown function of the variable $s\in \R$ based on the knowledge of the function for $s=n$, when $n$ is a positive integer. In the work of Brunet and Derrida \cite{brunet2000ground, brunet2000probability}, they consider together the series expansion in $n$ and the series expansion in $c$, the strength of the noise, and notice that the coefficients in $c$ are polynomials in $n$, hence analytic functions of $n$. The resummation of these coefficients as a power series in $c$ yields the conjectural cumulant generating function.  
Moreover, the approach relies on the completeness of Bethe ansatz eigenfunctions for the Lieb-Liniger Hamiltonian $H_n$ defined in \eqref{eq:defHn}. Finally, in our computation of the ground state energy, following \cite{brunet2000ground, brunet2000probability}, we pick certain solutions of functional equations which admit several solutions. The solutions that we pick are the simplest, in a certain sense, and we believe that these are the ones corresponding to the ground state energy while other solutions should correspond to excited states.

\medskip 
In order to obtain mathematical proofs of the results derived in the present paper, it may be useful to consider periodic or two-sided boundary models for which a richer mathematical structure is available. Promising candidates are periodic Schur processes \cite{borodin2007periodic}, free boundary Schur processes \cite{betea2018free} and their Macdonald analogues \cite{he2023periodic}. These models of random partitions are related to growth processes on bounded domains with either periodic or two sided boundary conditions \cite{betea2019periodic, betea2019new, betea2022peaks}.  It would be particularly interesting to compare cumulant generating functions obtained as the top eigenvalue of deformed Markov matrices and limits of $\log \mathbb E\left[ e^{z\lambda_1} \right]$ where $\lambda_1$ is the largest part of a partition $\lambda$ distributed according to the  periodic or free boundary Schur or Macdonald measure.

\section{Limit of the large-time cumulants of ASEP}
\label{sec:limitLazarescu}

\subsection{Definition of the model}

Let us first recall the definition of the asymmetric exclusion process (ASEP) with open boundaries. We will use here the notations of 
\cite{LazarescuThesis,gorissen2012exact} 
with the only change that $L$ there is denoted $\ell$ here. The system is a one dimensional lattice gas on $\ell$ sites. At any given time, 
each site $1 \leq i \leq \ell$ is either empty or occupied by at most one particle. A configuration of particles is described by a collection of occupation numbers $( \tau_i )_{1\leq i\leq \ell}$ where $\tau_i=1$ when site $i$ is occupied and 
$\tau_i=0$ when site $i$ is empty. The model depends on the bulk parameter $q$ and on the boundary parameters $\alpha, \beta, \gamma, \delta$. 
The dynamics is a continuous time Markov process $\tau(t)$ on the space state $( \tau_i )_{1\leq i\leq \ell} \in \{0,1\}^\ell$ defined as follows: 
 
At any given time $t>0$ and for any $i$, $1 \leq i \leq \ell-1$, each particle jumps from $i$ to $i+1$ with exponential rate $\tau_{i}(1-\tau_{i+1})$
and from $i+1$ to $i$ with exponential rate $q \tau_{i+1}(1-\tau_{i})$. In addition a particle at site $1$ is created 
with exponential rate $\alpha (1- \tau_1)$ and annihilated with exponential rate $\gamma \tau_1$, and a particle at site $\ell$ is created 
with exponential rate $\delta (1- \tau_\ell)$ and annihilated with exponential rate $\beta \tau_\ell$. All these events are independent.

Physically this is equivalent to consider reservoirs at both ends with imposed densities $\rho_a, \rho_b$ satisfying the mean current conservation conditions 
$j_a=(1-q) \rho_a (1-\rho_a)=\alpha(1-\rho_a)-\gamma \rho_a$ and 
$j_b=(1-q) \rho_b (1-\rho_b)=\beta \rho_b -\delta (1-\rho_b)$ \cite{derrida1993exact}. Hence, it is convenient to introduce parameters 
\begin{align}
    a&= \frac{1}{2\alpha} \left( (1-q-\alpha-\gamma) + \sqrt{(1-q-\alpha-\gamma)^2+4\alpha\gamma} \right)\\
    \widetilde a&= \frac{1}{2\alpha} \left( (1-q-\alpha-\gamma) - \sqrt{(1-q-\alpha-\gamma)^2+4\alpha\gamma} \right)\\
    b&=\frac{1}{2\beta} \left( (1-q-\beta+\delta) + \sqrt{ (1-q-\beta+\delta)^2+4\beta\delta} \right)\\
    \widetilde b&=\frac{1}{2\beta} \left( (1-q-\beta+\delta) - \sqrt{ (1-q-\beta+\delta)^2+4\beta\delta} \right)
\end{align}
so that 
\be 
\rho_a=\frac{1}{1+a}, \;\;\;\rho_b=\frac{b}{1+b}. 
\ee 
In other terms, we have parametrized the jump rates as 
\bea  
&&     \alpha= \frac{(1-q)}{(1+a)(1+\widetilde a)} 
 , \quad  \quad    \gamma= -\frac{a\widetilde a(1-q)}{(1+a)(1+\widetilde a)} \\ 
&&          \beta= \frac{(1-q)}{(1+b)(1+\widetilde b)} , \quad  \quad  
     \delta = -\frac{b\widetilde b(1-q)}{(1+b)(1+\widetilde b)}
\eea 
We will sometimes restrict to parameters $\alpha, \beta, \gamma, \delta$ solving Liggett's condition \cite{liggett1975ergodic}
\begin{equation}
   \alpha+\frac{\gamma}{q}=1, \;\;\; \beta+\frac{\delta}{q}=1,
    \label{eq:Liggett}
\end{equation}
which corresponds to setting $\widetilde a=\widetilde b=-q$ so that $\alpha=\rho_a=\frac{1}{1+a}$ and $\delta = q\rho_b=\frac{q b}{1+b} $. 

We may slightly extend the definition of ASEP by considering it as a Markov process on $\lbrace 0,1\rbrace^{\ell}\times \Z$, described by occupation numbers $\tau\in \lbrace 0,1\rbrace^{\ell}$, as well as the net number $N_t\in \Z$ of particles that have entered the system, from the left reservoir to site $1$, minus the number of particles that have exited the system from site $1$ to the left reservoir since time $t=0$.  On this extended state space, we define a height function $h_t(i)$, following \cite{corwin2018open}, as 
\begin{equation}
    h_t(i) = -2N_t+ \sum_{j=1}^i (2\tau_j-1), \;\;\; h_t(0)=-2N_t. 
    \label{eq:defhaightfunction}
\end{equation}

\begin{remark}
The parametrization involving $a, \widetilde a, b, \widetilde b$ is also convenient to apply the matrix product ansatz. One may chose matrices $D=1+d$, $E=1+e$, where $de-qed=1-q$ that is the $q$-oscillator algebra. Explicitly, one may chose $d=\sum_{n=1}^{\infty} (1-q^n) \ket{n-1}\bra{n}$ and $e=\sum_{n=1}^{\infty} \ket{n+1}\bra{n}$. With this choice, the boundary vectors are simply chosen \cite{sasamoto1999one} as $\bra{W} = \sum_{n=1}^{\infty} W_n \bra{n}$ and $\ket{V} = \sum_{n=1}^{\infty} V_n \ket{n}$ with 
\be 
W_n = H_n(a,\widetilde a), \;\;\; V_n = \frac{H_n(b,\widetilde b)}{(q)_n},
\ee 
where $H_n$ are $q$-Hermite polynomials
\be 
H_n(x,y) = \sum_{k=0}^n \frac{(q)_n}{(q)_k(q)_{n-k}}x^ky^{n-k}
\ee 
 We refer to \cite[chapter 2]{LazarescuThesis} for more references and details. It is interesting to note that the same boundary weights also arise in the half-space Hall-Littlewood process \cite{he2024boundary}.
\end{remark}

\subsection{Large time cumulants of open ASEP}

We now recall the results of \cite{LazarescuThesis,gorissen2012exact}. 
We define the function 
\be \mathcal E(\mu)=\lim_{t\to\infty} \frac{1}{t} \log \mathbb E\left[ e^{\mu N_t}\right]
\ee
where we recall that $N_t$ is defined just above \eqref{eq:defhaightfunction}. The function $\mathcal E(\mu)$ is computed as the principal eigenvalue of a deformation of the Markov transition matrix of the system, where the matrix elements corresponding to a transition increasing $N_t$ by one (resp. decreasing $N_t$ by one) are multiplied by $e^{\mu}$ (resp. $e^{-\mu}$). 

In terms of the ASEP height function \eqref{eq:defhaightfunction}, 
\be 
\mathcal E(\mu)=\lim_{t\to\infty} \frac{1}{t} \log \mathbb E\left[ e^{-\mu h_t(0)/2}\right].
\label{eq:defcalE}
\ee 

We start from the equations in Section III.3 of \cite{LazarescuThesis}, more precisely  \cite[Eqs III.64 to III.69]{LazarescuThesis}. Let us define the function 
\be 
F_{\ell}(z) = \frac{\left( (1+z)(1+z^{-1}) \right)^\ell (z^2, z^{-2})_{\infty} }{(az, az^{-1}, \widetilde a z, \widetilde az^{-1},bz, bz^{-1}, \widetilde b z, \widetilde bz^{-1})_{\infty}}
\label{eq:defFell}
\ee
where $(x)_{\infty}=\prod_{i=0}^{\infty} (1-q^ix)$ and $(x_1,\dots,x_k)_{\infty} = (x_1)_{\infty} \dots (x_k)_{\infty}$.  This function arises in the study of the open ASEP steady-state using Matrix Product Ansatz \cite{sasamoto1999one}. 

The expression for the principal eigenvalue $E(\mu)$ is obtained using a deformation of the Matrix Product Ansatz \cite{gorissen2012exact}. It  takes a form similar to the case of periodic boundary conditions \cite{prolhac2010tree}. One introduces the kernel 
\be 
{\cal K}(z,\widetilde z) = 2 \sum_{k=1}^{\infty} \frac{q^k}{1-q^k} \left( \left(\frac{z}{\widetilde z}\right)^k + \left(\frac{\widetilde z}{z}\right)^k \right) 
= 2 \sum_{k=1}^{\infty} \frac{q^k z}{\widetilde z - q^k z} + \frac{q^k \widetilde z}{z - q^k \widetilde z}
\ee 
and the convolution operator
\bea \label{convolu} 
X[f](z) = \oint \frac{d \widetilde z}{2 \I \pi \widetilde z} \mathcal K(z,\widetilde z) f(\widetilde z),
\eea  
where the contour is a positively oriented circle around $0$ of radius $1$ (this requires that $\vert a\vert, \vert \widetilde a\vert, \vert b\vert, \vert \widetilde b\vert <1$, otherwise the contour must be chosen differently).  Then, we define a function $W(z)$ solving the functional equation 
\be 
W(z) = - \frac{1}{2} \log( 1 - B F(z) e^{X[W](z)} ).
\label{eq:functionalequationASEP}
\ee 
 Finally, the function $\mathcal E(\mu)$ is defined by eliminating $B$ in the system of equations 
 \begin{align}    \label{mueq} 
     \mu& =-\oint \frac{dz}{2\I\pi z} W(z)\\
     \mathcal E(\mu) &= -(1-q) \oint \frac{dz}{2\I\pi z} \frac{W(z)}{(1+z)(1+z^{-1})}. \label{Eeq} 
 \end{align}

\begin{remark} \label{rem:periodic}
The above formulae also apply for the periodic ASEP with $\ell$ sites and $N$ particles, up to some simple modifications.
This was worked out in \cite{prolhac2010tree}, see also \cite{LazarescuThesis} Section II.2.2.
and \cite{gorissen2012exact}. These modifications amount to substitute $W \to W_{\rm per}$ where $W_{\rm per}$ now obeys
\be \label{functasepperiodic}
W_{\rm per}(z) = - \log(1- B F_{\rm per}(z) e^{X[W_{\rm per}](z)})
\ee 
where $X[W]$ is defined as in \eqref{convolu} but with a kernel $\mathcal K_{\rm per}=\mathcal K/2$, and where
$F_{\rm per}(z)=(1+z)^\ell/z^N$.
\end{remark}

\subsection{Convergence of open ASEP to the KPZ equation} 
\label{sec:ASEPtoKPZ}
Let us now describe the convergence of ASEP to the KPZ equation proved in \cite{corwin2018open}. We consider the scalings, for $\eps>0$, 
\begin{equation} 
q=e^{-\eps}, \quad  \quad i= 4 x\eps^{-2}, \quad \quad \ell= 4  L \eps^{-2}. 
\label{eq:scalings}
\end{equation}
We scale boundary  parameters as 
\be 
a=q^u, \;\; \widetilde a=-q, \;\;b=q^{v}, \;\; \widetilde b=-q,
\label{eq:scalingsboundary}
\ee 
so that the boundary rates  $\alpha, \beta, \gamma, \delta$ are scaled as 
\be 
\rho_a = \frac{1}{2} + \frac{u}{4}  \eps + o(\eps), \quad  \quad \rho_b = \frac{1}{2} - \frac{v}{4}  \eps+o(\eps).
\ee 
For the moment, we assume that $\rho_a>\rho_b$, so that $ab<1$ and $u+v>0$. This corresponds to the maximal current phase.  Furthermore, we assume for the moment, for  technical convenience, that $u,v>0$ such that $\vert a\vert , \vert b\vert <1$ and we can apply the integral formulas from the previous paragraph without changing integration contours. 

Then, we define the discrete Hopf-Cole transform\footnote{we define the Hopf-Cole transform in a slightly different manner than \cite{corwin2018open}. We consider the usual Hopf-Cole transform $q^{-h_{T}(j)}e^{\nu T}$ with $\nu=(\sqrt{p}-\sqrt{q})^2$  with $T=\frac{t}{2\sqrt{q}}$. This ensures that our $Z_t(j)$ has the same definition as $Z_t(j)$ in \cite[Definition 2.12]{corwin2018open}.}
\be 
Z_t(j) = q^{-h_{\frac{t}{2\sqrt{q}}}(j)/2}e^{\left(1-\sqrt{q}\right)^2 \frac{t}{2\sqrt{q}} }. 
\ee 
 Theorem 2.18 in \cite{corwin2018open} states that as $\eps$ goes to zero, 
 \be Z_{16 \eps^{-4} T}(4\eps^{-2} x) = Z(x,T)+o(1)
 \label{eq:discreteHopfColetoSHE}
 \ee  
  where $Z$ is a solution of the SHE equation \eqref{eq:SHE} with boundary parameters $u$ and $v$. Strictly speaking, \cite{corwin2018open} has restrictions on the range of boundary parameters, and those restrictions are removed in \cite{parekh2019kpz}. Moreover, the convergence \eqref{eq:discreteHopfColetoSHE} is proved in the sense  of weak convergence of the left-hand-side to the right-hand-side in the space of continuous functions, but the notation used in \eqref{eq:discreteHopfColetoSHE} will be convenient for later purposes. 
To compare \eqref{eq:discreteHopfColetoSHE} with the statement in \cite{corwin2018open}, one should relate the parameter $N$ there to our parameter $\eps$ by $N=4\eps^{-2}$. Note that the parameter $\eps$ in \cite{corwin2018open} has a different meaning. 

Letting $T=\tau \sqrt{q}$ and taking the logarithm of both sides in \eqref{eq:discreteHopfColetoSHE}, we obtain that 
\be 
\frac{\eps}{2} h_{8\eps^{-4}\tau}(4\eps^{-2}x) + 8 \eps^{-4}(1-e^{-\eps/2})^2 \tau = \frac{\eps}{2 }h_{8\eps^{-4}\tau}(4\eps^{-2}x) + 2\tau \eps^{-2} -\tau \eps^{-1} + \frac{7 \tau}{24} +O(\eps) 
\ee 
converges to $h(x,\tau \sqrt{q})=h(x,\tau) + O(\eps)$ where $h$ is the solution of the KPZ equation on the interval \eqref{eq:KPZ}. Note that this scaling is consistent with \cite[Definition 7.5]{barraquand2018stochastic}. 

To conclude, we have obtained that 
\be 
h_{8\eps^{-4}\tau}(0) = 2\eps^{-1}\left( h(0,\tau) -2\tau\eps^{-2} +\tau\eps^{-1} -\frac{7\tau}{24} + o(1) \right).
\label{eq:conclusionscalingh}
\ee 

\subsection{Derivation of the main result}

With $q= e^{-\epsilon}$ and the change of variable
$y= k \epsilon$, $z= e^{\epsilon w}$, $\widetilde z= e^{\epsilon \widetilde w}$,  one obtains
\be 
\mathcal K(z,\widetilde z) -4  \sum_{k=1}^{\infty} \frac{q^k}{1-q^k} 
\sim \frac{2}{\epsilon} \int_0^{+\infty}  dy \frac{e^{-y}}{1 - e^{-y}} (e^{y (w-\widetilde w)} + e^{- y (w-\widetilde w) } - 2) 
\ee 
Using 
\be 
\int_0^{+\infty}  dy \frac{1}{e^y - 1} (e^{y w} -1) = - \gamma_E - \psi(1-w)  
\ee 
One finds
\be 
\mathcal K(z,\widetilde z) -4  \sum_{k=1}^{\infty} \frac{q^k}{1-q^k} 
\sim \frac{2}{\epsilon} (-2 \gamma_E - \psi(1- w+ \widetilde w) - \psi(1+ w - \widetilde w) ).
\ee 
Now we consider the limit of $F_\ell(z)$ defined in \eqref{eq:defFell}. We have 
\be 
(q^x)_{\infty} = \frac{(1-q)^{1-x} (q)_{\infty} }{\Gamma_q(x)} 
\label{eq:qGamma}
\ee
where $\Gamma_q(x)$ converges to $\Gamma(x)$ as $q\to 1$. 
Further, using \cite[Proposition 2.3]{corwin2024stationary}
\be 
(-q^w)_{\infty} = e^{\frac{\pi^2}{12\eps} -(w-1/2)\log 2 + O(\eps)}
\ee 
so that there exists a constant $C_1(\eps)$ such that 
\be 
\frac{(\widetilde az)_{\infty}(\widetilde a/z)_{\infty}(\widetilde bz)_{\infty}(\widetilde b/z)_{\infty}}{C_1(\eps)} \xrightarrow{\eps\to 0}
1
\label{eq:convergencejunk}\ee
Hence, under the scalings \eqref{eq:scalings} and the same change of variables $z= e^{\epsilon w}=q^{-w}$ as above, \eqref{eq:qGamma} and \eqref{eq:convergencejunk} ensure that there exists $C(\epsilon,u,v)$ such that 
\be 
\frac{F_\ell(z)}{C(\epsilon,u,v) } \xrightarrow{\eps\to 0}  4^{\ell} \Psi_L(w)
\ee 
where $\Psi_L(w)$ is the same function as $\Psi(w)$ defined in \eqref{defPsi} where we have made explicit the $L$-dependence.

Now we examine the limit of the functional equation \eqref{eq:functionalequationASEP}. Letting $2\zeta=B C(\eps, u,v)4^{\ell}$ the functional equation \eqref{eq:functionalequationASEP} becomes \eqref{eq:functionalequation}. It is convenient to rename the function $z\mapsto W(z)$, which depends on a free parameter $B$, as a function $w\mapsto U(w)$, depending on a free parameter $\zeta$. Let us now consider the limit of \eqref{mueq} and \eqref{Eeq}. Using \eqref{mueq}, we find that 
 $\mu\sim-\eps s $ with 
 \be s= \int_{\I\R} \frac{dw}{2\I\pi} U(w).\ee 
 Then, we have that 
\begin{align}
  \mathcal E(\mu) &=   - (1-q) \int \frac{dz W(z)}{2 \I \pi z (1+z)(1+z^{-1})}  \\ 
  &= - (1-q) \int \frac{dz}{2 \I \pi z} \frac{F_{\ell-1}(z)}{F_\ell(z)} W(z)\\
  &= -\eps(1-e^{-\eps}) \int \frac{dw}{2 \I \pi} \frac{1}{4} \left( 1 - \frac{\eps^2}{4} \partial_L \log \Psi_L(w) \right) 
U(w) +o(\eps^4)\\
&=- \left( \frac{\epsilon^2}{4} -\frac{\eps^3}{8}+\frac{\eps^4}{24}\right) \int \frac{dw}{2 \I \pi} U(w) 
+ \frac{\epsilon^4}{16}  \int_{\I \mathbb{R}} \frac{d w}{ 2 \I \pi} w^2 U(w) +o(\eps^4)\\ 
&= - s \left( \frac{\epsilon^2}{4} -\frac{\eps^3}{8}+\frac{\eps^4}{24}\right) 
+ \frac{\epsilon^4}{16}  \int_{\I \mathbb{R}} \frac{d w}{ 2 \I \pi} w^2 U(w)+o(\eps^4)
\label{eq:expansioncalE1}
\end{align}
On the other hand, by the definition of $\mathcal E(\mu)$ in \eqref{eq:defcalE}, plugging the scaling relation in \eqref{eq:conclusionscalingh}, we obtain 
\begin{align}
\mathcal E(\mu) &= \lim_{\tau\to \infty} \frac{1}{8\eps^{-4}\tau} \log \mathbb E e^{s \left(h(0,\tau)  -2 \eps^{-2}\tau+ \eps^{-1} \tau - \frac{7\tau}{24} +o(1) \right)  } \\
&= \frac{\eps^{4}}{8} \left(  E(s) +  s\left(-2\eps^{-2} +\eps^{-1} -\frac{7}{24}\right) \right)+ o(\eps^4)
\label{eq:expansioncalE2}
\end{align}
Combining \eqref{eq:expansioncalE1} and \eqref{eq:expansioncalE2}, we obtain that 
\be 
E(s) = \frac{-s}{24} + \frac{1}{2}\int_{\I\R} \frac{dw}{2\I\pi} w^2U(w).
\ee 
Hence we have shown the results announced in \eqref{subsubsec:functional}.

\begin{remark}
For some special values of $u$ and $v$, the function  $\Psi(u)$ simplifies. For instance,
for $u=0, v=1/2$  one has $\Psi(w)=4\pi e^{w^2 L}$ so that   the cumulants
can be related to the periodic case for which $\Psi_{\rm per}(w)=e^{w^2 L/2}$. 
More generally, for $v=u+1/2$ and $u=2 k$ where $k \in \mathbb{N}$ one has $\Psi(u)= Q_k(w^2) e^{w^2 L}$
where $Q_k$ is a polynomial of degree $k$. 
\end{remark}
\begin{remark}

The cumulant generating function  $\mathcal E(\mu)$ of ASEP satisfies Gallavotti-Cohen symmetry \cite{lebowitz1999galavotti}  (see \cite[Section I.2.3 and III.1.2]{LazarescuThesis} for details): 
$$ \mathcal E(\mu)  = \mathcal E(\mu_0-\mu), \; \text{ with } \mu_0=-\log \frac{\alpha\beta}{\gamma\delta q^{\ell-1}}.$$
In the scaling limit \eqref{eq:scalings},  \eqref{eq:scalingsboundary}, and $\mu\sim -\eps s$, one finds that $\mu_0 =  4L \eps^{-1}+ O(\eps)$, hence the symmetry is lost in the limit to the KPZ equation.   
\end{remark}

\section{Formulas for the cumulants}
\label{sec:cumulants}
\subsection{General principle} 
Our goal is to extract the coefficients $c_k$ of the cumulant generating series $E(s)$
\be 
 E(s)=\sum_{k\geq 1} \frac{s^k}{k!} c_k.
 \label{eq:series1}
\ee
The general principle is that given the equations \eqref{systemEmu1}, \eqref{systemEmu2}, \eqref{eq:functionalequation}, we first 
extract the series coefficients of $E$ and $s$ in the variable $\zeta$, and then combinatorially retrieve the coefficients of $E$ 
in the variable $s$. These coefficients are expressed as simple integrals involving the coefficients of the function $U(w)$ in powers of $\zeta$. 

More precisely, we set 
\be 
U(w) = \sum_{k \geq 1} U_k(w) \zeta^k,\;\; s = \sum_{k\geq 1} s_k \zeta^k \;\;\; E = \sum_{k\geq 1} \eps_k \zeta^k.
\ee
so that \eqref{systemEmu1} and \eqref{systemEmu2} imply that 
\be 
s_k = \int_{\I\R} \frac{dw}{2\I\pi} U_k(w), \;\; \eps_k = \frac{1}{2}\int_{\I\R} \frac{dw}{2\I\pi} w^2 U_k(w).
\ee 
Let us assume for the moment that the $U_k(w)$ are known.

By plugging the  series expansion of $s$ in powers of $\zeta$ defined in \eqref{eq:relationcoeffs}
into \eqref{eq:series1}, and extracting the coefficient of $\zeta^k$ in $E$, we find that 
\be 
\eps_k = \sum_{1 \leq \ell \leq k} \frac{c_\ell}{\ell!} d_{k,\ell} 
\label{eq:epsilonintermsofc}
\ee
where, for $\ell \leq k$
\be 
d_{k,\ell}= \sum_{\begin{matrix}(i_1, \dots, i_{\ell})\in \Z_{\geq 1}^{\ell} \\ i_1+\dots+i_{\ell}=k\end{matrix} } \prod_{j=1}^{\ell} s_{i_j}.
\ee 
and $d_{k,\ell}=0$ for $\ell > k$.
Using  notations $\vec C = (c_i/i!)_{i\geq 1}, \vec \eps = (\eps_i)_{i\geq 1}, D=(d_{i,j})_{i,j\geq 1}$, 
we have the equation $\vec\eps = A\vec C$, hence formally we find $\vec C = D^{-1}\vec \eps$. 

The inverse of the triangular matrix $D$ does not have a particularly nice form. However, since the diagonal coefficients are  $d_{kk}=s_1^k$, Equation \eqref{eq:epsilonintermsofc} implies that 
\be \frac{c_k}{k!} = \frac{1}{s_1^k} \left(\eps_k - \sum_{\ell=1}^{k-1} d_{k,\ell} \frac{c_\ell}{\ell!} \right), \ee
which allows to compute the $c_k$ recursively.

Let us give explicit expressions for the first few cumulants. The top $4\times 4$ principal submatrix of $D$ is
\be  \begin{pmatrix}
 s _1 & 0 & 0 & 0  \\
 s _2 & s _1^2 & 0 & 0  \\
 s _3 & 2 s _1 s _2 & s _1^3 & 0  \\
 s _4 & s _2^2+2 s _1 s _3 & 3 s _1^2 s _2 & s _1^4 
\end{pmatrix}.\ee
Its inverse is 
\be
\begin{pmatrix}
 \frac{1}{s _1} & 0 & 0 & 0  \\
 -\frac{s _2}{s _1^3} & \frac{1}{s _1^2} & 0 & 0  \\
 \frac{2 s _2^2-s _1 s _3}{s _1^5} & -\frac{2 s _2}{s _1^4} & \frac{1}{s _1^3} & 0  \\
 -\frac{5 s _2^3-5 s _1 s _3 s _2+s _1^2 s _4}{s _1^7} & \frac{5 s _2^2-2 s _1 s _3}{s _1^6} & -\frac{3
   s _2}{s _1^5} & \frac{1}{s _1^4}  
\end{pmatrix}.\ee
Hence, 

\begin{align}
    c_1 &= \frac{-1}{24} + \frac{\eps_1}{s_1}\\  
    c_2&=  2 \frac{ \eps_2 s_1 -\eps_1 s_2}{s_1^3}\\ 
    c_3 &= 3! \frac{\eps_3 s_1^2 - 2 \eps_2 s_2s_1 + \eps_1(2s_2^2 - s_1s_3) }{s_1^5}\\  
    c_4&= 4! \frac{\eps_4 s_1^3 -3 \eps_3 s_2 s_1^2 + \eps_2 s_1(5s_2^2 - 2s_1s_3) - \eps_1 (5 s _2^3-5 s _1 s _3 s _2+s _1^2 s _4)}{s_1^7}\\ 
\end{align}
It will be convenient to use renormalized coefficients
\be 
m_i = \frac{s_i}{s_1}, \;\;\; e_i=\frac{\eps_i}{s_1}.
\ee 
Then, we have 
\begin{subequations}
    \begin{align}
    c_1 &= \frac{-1}{24} + e_1 \\  
    c_2&=  2 \frac{ e_2  -e_1 m_2}{s_1}\\ 
    c_3 &= 3! \frac{e_3  - 2 e_2 m_2 + e_1(2m_2^2 - m_3) }{s_1^2}\\  
    c_4&= 4! \frac{e_4 -3 \e_3 m_2  + e_2 (5m_2^2 - 2m_3) - e_1 (5 m_2^3 - 5 m_3 m_2+ m_4)}{s_1^3}
\end{align}
\label{eq:firstcumulantscoeffs}
\end{subequations}

\begin{remark}
It is possible to write explicitly the coefficient $c_k$ for arbitrary $k$ in terms of the $s_i$ and $\eps_j$. Indeed, using Lagrange inversion formula, one can find coefficients $\zeta_i$ such that the map $s\mapsto \zeta(s)$,   the inverse of the map $\zeta\mapsto s(\zeta)$, is written in power series as $\zeta=\sum_{k\geq 1} \zeta_k s^k$. The coefficient $\zeta_k$ can be explicitly written in terms of  $(s_j)_{1\leq j\leq k}$ via incomplete Bell polynomials.  Then, using that $E(s) = \sum_{k\geq 1} \eps_k (\zeta(s))^k$, it is easy to extract the coefficient of $\zeta^k$, which is $c_k/k!$. In the context of the periodic TASEP, this explicit computation was first performed in \cite{prolhac2010tree}.
\end{remark}

\begin{remark}  
    One has the alternative formula 
    \be 
   \frac{1}{\ell!} d_{k,\ell} = 
   \sum_{\begin{matrix} m_1\geq 0,m_2 \geq 0,\dots \\ \sum_{j\geq 1} m_j = \ell \\
   \sum_{j\geq 1} j m_j = k \end{matrix}} \prod_{j \geq 1} \frac{s_j^{m_j}}{m_j!} 
    \ee 
    This leads to the generating function
    \be 
    \sum_{1 \leq \ell, k  } \frac{1}{\ell!} d_{k,\ell} x^k y^\ell = \exp\left( \sum_{j \geq 1} y x^j s_j \right) 
    \ee 
\end{remark}

\subsection{Computation of the coefficients $U_n$}
Recall that the coefficients $s_k$ and $\eps_k$ are expressed in terms of the function $U_k(w)$. In this section we explain how to  compute the coefficients $U_k(w)$. 
The functional equation \eqref{eq:functionalequation} yields the following expressions for the first few coefficients: 
\begin{subequations}
\begin{align} 
 U_1 &= \Psi \\
 U_2 &= \Psi^2 + \Psi ({\sf k} \Psi)  \\
 U_3 &= \frac{4}{3} \Psi^3 + 2 \Psi^2 ({\sf k} \Psi) + \frac{1}{2}  \Psi  ({\sf k} \Psi)^2 + \Psi ( {\sf k} \Psi^2 ) 
 + \Psi {\sf k}( \Psi ({\sf k} \Psi)) 
\end{align} 
\label{U_k}
\end{subequations}
More generally, recall that $U_n(w)$ denotes the coefficient of $\zeta^n$ in $U(w)$. Replacing $U(w)$ by $\frac{-1}{2} \log(1- 2 \zeta \Psi e^{\kk U})$, and using the series expansion of the logarithm and exponential functions, we obtain the following (below we denote by $[\zeta^n]f$  the coefficient of $\zeta^n$ when we expand the expression $f$ in power series in $\zeta$). 
\begin{align}
    U_n(w) &= [\zeta^n]  \frac{1}{2} \sum_{\ell \geq 1} \frac{1}{\ell} \left( 2\zeta\Psi e^{\kk U} \right)^{\ell} \\ 
    &= [\zeta^n]\sum_{\ell \geq 1} \frac{2^{\ell-1} \zeta^\ell}{\ell} \Psi^\ell  e^{\ell \kk  U}\\ 
     &= [\zeta^n]\sum_{\ell \geq 1} \frac{2^{\ell-1} \zeta^\ell}{\ell} \Psi^\ell  \sum_{i=0}^{\infty} \frac{\ell^i}{i!} (\kk  U)^i\\ 
     &= [\zeta^n]\sum_{\ell \geq 1} \frac{2^{\ell-1} \zeta^\ell}{\ell} \Psi^\ell  \sum_{i=0}^{\infty} \frac{\ell^i}{i!} \left(\sum_{p\geq 1} \zeta^p \kk  U_p  \right)^i\\ 
    &= [\zeta^n]\sum_{\ell \geq 1} \frac{2^{\ell-1} \zeta^\ell}{\ell} \Psi^\ell  \sum_{i=0}^{\infty} \frac{\ell^i}{i!}  
       \sum_{p_1, \dots, p_i\geq 1} \prod_{j=1}^i  \zeta^{p_j} \kk  U_{p_j}.
       \label{eq:recurrenceforUbis}
\end{align}
Then, the coefficient of $\zeta^n$ must be extracted with care. It is more convenient to deal separately with the term $\ell=n$, for which the sum over $i$ gives $1$ (i.e. the first term in the expansion of the exponential), and the term $\ell<n$, for which we must take $i>0$. We find that  $U_n(w)$ can be computed by recurrence as 
\be \label{recursionU}
U_n = \frac{2^{n-1}}{n} \Psi^n + \sum_{\ell =1}^{n-1}   \frac{2^{\ell-1} }{\ell} \Psi^\ell  \sum_{i=1}^{n-\ell} \frac{\ell^i}{i!}  \sum_{\begin{matrix} (p_1, \dots, p_i) \in \mathbb Z_{\geq 1}^i \\ p_1+\dots+p_i=n-\ell \end{matrix} } \prod_{j=1}^i \kk  U_{p_j}(w),
\ee 
with the initial condition $U_1=\Psi$, as claimed in \eqref{eq:recurrenceforU}.

\begin{remark} The expression for $U_n$ obtained after solving the recursion \eqref{eq:recurrenceforU} is a sum of terms  of the form 
\be 
\Psi^{c} \prod_{j=1}^r \kk( F_j )   
\label{eq:treeexpansion}
\ee 
for some integers $c\geq 1$ and $r\geq 0$, where the $F_j$ are expressions involving $\Psi$ and $\kk$ (of lower homogeneity  in $\Psi$). To each such expression we can associate a rooted tree with vertices labeled by positive integers. We associate to $\Psi^c$ the tree with a unique vertex labeled by $c$. More generally, to an expression written as  \eqref{eq:treeexpansion} we associate the tree with a root labeled by $c$ and connected to $r$ subtrees corresponding to the expressions $F_j$, we denote those subtrees by $\mathcal T(F_j)$. Thus, the tree corresponding to \eqref{eq:treeexpansion} is the following: 
\begin{center}
\tikzstyle{arbre}=[rectangle,draw,thick]
    \begin{tikzpicture}[xscale=2, thick]
\node[circle, draw] (root) at (0,0) {$c$};
\node[arbre] (f1) at (-1.5,-2) {$\mathcal T(F_1)$};
\node[arbre] (f2) at (-0.5,-2) {$\mathcal T(F_2)$};
\draw (0.5,-2) node {$\ldots$};
\node[arbre] (f3) at (1.5,-2) {$\mathcal T(F_r)$};
\draw (root) -- (f1);
\draw (root) -- (f2);
\draw (root) -- (f3);
\end{tikzpicture}
\end{center}
For example, under these conventions, 
\be 
U_3 =  
\frac{4}{3}\raisebox{-1ex}{
\begin{tikzpicture}[xscale=2, thick] \node[circle, draw] (root) at (0,0) {$3$};\end{tikzpicture}} 
+ 2 \;\;\raisebox{-4ex}{\begin{tikzpicture}[xscale=2, thick] \node[circle, draw] (root) at (0,0) {$2$}; 
\node[circle, draw] (c) at (0,-1){$1$};
\draw (root) -- (c);
\end{tikzpicture}}
+ \frac{1}{2} \;\; \raisebox{-4ex}{\begin{tikzpicture}[xscale=1, thick] \node[circle, draw] (root) at (0,0) {$1$}; 
\node[circle, draw] (c1) at (-1,-1){$1$};
\node[circle, draw] (c2) at (1,-1){$1$};
\draw (root) -- (c1);
\draw (root) -- (c2);
\end{tikzpicture}}
+ \;\;\raisebox{-4ex}{\begin{tikzpicture}[xscale=2, thick] \node[circle, draw] (root) at (0,0) {$1$}; 
\node[circle, draw] (c) at (0,-1){$2$};
\draw (root) -- (c);
\end{tikzpicture}}
+ \;\; \raisebox{-8ex}{\begin{tikzpicture}[xscale=1, thick] \node[circle, draw] (root) at (0,0) {$1$}; 
\node[circle, draw] (c1) at (0,-1){$1$};
\node[circle, draw] (c2) at (0,-2){$1$};
\draw (root) -- (c1);
\draw (c1) -- (c2);
\end{tikzpicture}}\,.
\ee 
The occurrence of a tree expansion is not surprising. Similar tree expansion were originally found in the work \cite{prolhac2010tree}. 
\end{remark}

\begin{remark}
    In terms of the kernel $\bar \kk$ defined in \eqref{defkbar} the functional equation reads
    \be 
    \zeta \Psi e^{\bar \kk U} = \sinh(U) 
    \label{eq:functionalequationkbar}
    \ee 
    This alternatice form of the functional equation is convenient to compute the first coefficients in terms of $\bar\kk$. One gets
    \be \label{U2U3} 
    U_2= \Psi \bar \kk \Psi \quad , \quad U_3 = - \frac{1}{6} \Psi^3 + \frac{1}{2} \Psi  (\bar \kk \Psi)^2 + \Psi \bar \kk (\Psi \bar \kk \Psi)
    \ee 
    Furthermore one can also define
    \be 
    \zeta' = \zeta s_1, \quad  \quad U = \frac{ \Psi V}{s_1} = \hat \Psi V, \quad  \quad s_1= \int_{\I \mathbb{R}} \frac{dw}{2 \I \pi} \Psi(w),
    \ee  
    and the functional equation takes the form
    \be 
    \zeta'  e^{\bar \kk * V} = \frac{\sinh(\hat \Psi V)}{\hat \Psi}, \quad  \quad (\bar \kk * f)(w) = \int \nu(dw) \bar \kk(w,w') f(w') 
    \ee  
    Hence we obtain
    \bea  
    V = \sum_k (\zeta s_1)^k V_k, \quad  \quad V_1=1, \quad  \quad V_2= \bar \kk * 1, 
    \quad  \quad V_2= \bar \kk * (\bar \kk * 1)  + \frac{1}{2} (\bar \kk * 1 )^2 - \frac{1}{6}  \hat \Psi^2 
    \eea  
    One obtains
    \be 
    \frac{s_k}{s_1^k} 
    = \int_{\I \mathbb{R} } \nu(dw)  V_k(w), 
    \quad  \quad \frac{e_k}{s_1^k} =   \int_{\I \mathbb{R} } \nu(dw)  w^2 V_k(w)
    \ee 
\end{remark}
\begin{remark}
 Using the functional equation   \eqref{eq:functionalequationkbar}, we obtain an alternative recursion formula for $U_n(w)$. 
 We find that 
 \begin{multline} 
U_n = \gamma_{(n-1)/2} \mathds{1}_{n \text{ is odd}} \Psi^n \\ + \sum_{r =0}^{\lfloor \frac{ n-1}{2}\rfloor -1}  \gamma_{r} \Psi(w)^{2r+1} \sum_{i=1}^{n-2r +1} \frac{(2r+1)^i}{i!}  \sum_{\begin{matrix} (p_1, \dots, p_i) \in \mathbb Z_{\geq 1}^i \\ p_1+\dots+p_i=n-2r-1 \end{matrix} } \prod_{j=1}^i \kk  U_{p_j}(w),
\end{multline}  
where $\gamma_r= \frac{(-1)^r }{2^{2r}(2r+1)}\binom{2r}{r}$.
\end{remark}

\subsection{First few cumulants}
Using the notation for the measure $\nu(dw)$ in \eqref{eq:defnu}, we have 
\be 
m_k = \int_{\I\R} \nu(dw) \, \frac{U_k(w)}{\Psi(w)}, \;\;\;\;\;\; e_k = \frac{1}{2}\int_{\I\R} \nu(dw) \, w^2 \frac{U_k(w)}{\Psi(w)}.
\ee 
\subsubsection{First cumulant}
Using \eqref{eq:firstcumulantscoeffs} and \eqref{U_k}, we obtain the first cumulant as
\be 
c_1 = e_1 - \frac{1}{24} = - \frac{1}{24} + \frac{1}{2}  \int_{\I\R} \nu(dw) w^2.
\ee

\subsubsection{Second cumulant} Using \eqref{eq:firstcumulantscoeffs} and \eqref{U_k} we have 
\begin{align}
 c_2 &= \frac{2}{s_1} ( e_2-e_1 m_2 ) \\
 &= \frac{1}{\int_{\I \mathbb{R}} \frac{d w}{ 2 \I \pi}  \Psi(w)} \left( 
 \int \nu(dw) w^2 (\Psi(w) + ({\sf k} \Psi)(w) ) - \left(\int \nu(dw') (w')^2\right) \int \nu(dw)  (\Psi(w) + ({\sf k} \Psi)(w) ) \right) \nn \\
 &= \frac{1}{\int_{\I \mathbb{R}} \frac{d w}{ 2 \I \pi}  \Psi(w)} \int \nu(d w)  ( ({\sf k} \Psi)(w) +\Psi(w) ) \left( w^2 -\int_{\I\R} (w')^2\nu(dw') \right)  
\end{align}   
This leads to the expression given above in \eqref{c2exact}, namely
\be \label{eq:secondcumulantbulk}
c_2 = 
\int_{\I\R} \nu(dw_1) \int_{\I\R} \nu(dw_2)  \bar\kk(w_1,w_2)  (w_1^2 - \int_{\I\R } \nu(dw) w^2 ),
\ee 
where we recall that 
\be
\bar \kk(w_1,w_2) = \kk(w_1,w_2)+\delta(w_1-w_2).
\ee
Furthermore using the identity \eqref{c2derivativeL} we can rewrite the second cumulant as 
\be 
c_2 =  \frac{1}{2} \partial_L \left( \int_{\I\R} \nu(dw_1) \int_{\I\R} \nu(dw_2) \bar \kk(w_1,w_2) \right).
\ee 
When $u,v$ go to zero, we find 
\be
c_2(0,0,L) = \frac{1}{2} \int_{\R} \frac{dx}{2 \pi} x^2 \Psi_{0,0}(I x) = 
\frac{1}{2} \int_{\R} \frac{dx}{2 \pi}  4 \pi x \coth(\pi x) e^{- L x^2} 
\label{eq:c200}
\ee 
where $\Psi_{00}(w) = \lim_{u,v\to 0} \Psi(w)$, as claimed in \eqref{eq:c200intro}. Indeed in this limit the dominant term is 
\be 
c_2 = - 
\int_{\I\R} \nu(dw_1) \int_{\I\R} \nu(dw_2)  \delta(w_1-w_2)   \int_{\I\R } \nu(dw) w^2 
= - \frac{\int \frac{dw}{2 \pi}  \Psi(w)^2 }{(\int \frac{dw}{2 \pi} \Psi(w))^2 }
\frac{\int \frac{dw}{2 \pi}  \Psi(w) w^2 }{(\int \frac{dw}{2 \pi} \Psi(w) }
\ee 
All integrals are divergent and dominated by $w \sim u$ (their divergent parts being then easily computed), except 
$\int \frac{dw}{2 \pi}  \Psi(w) w^2$ where one can set $u=v=0$.

\subsubsection{Third cumulant} 
We note that the final results for the cumulants can be expressed in terms of $\nu(dw)$, which is a positive measure for $w$ on the imaginary axis 
and normalized to unity. It can thus be interpreted as a probability measure. It is useful in what follows to 
introduce the following notation, for any function of one or several variables $f(w)$, $f(w_1,w_2)$, etc
\be 
\langle f(w) \rangle = \int_{\I \mathbb{R} } \nu(dw) f(w),  \quad  \quad 
\langle f(w_1,w_2) \rangle = \int_{\I \mathbb{R} } \nu(dw_1) \int_{\I \mathbb{R} } \nu(dw_2)  f(w_1,w_2)  
\ee
for the expectation value w.r.t. that measure (the variables $w_i$ being considered as i.i.d. independent) 

Then one has
\bea  
&& c_3 = \frac{6}{s_1^2} (e_3- 2 e_2 m_2 + 2 e_1 m_2^2 - e_1 m_3) \\
&& = \frac{3}{s_1^2}  \left(   \left\langle w^2 \frac{U_3}{\Psi} \right\rangle - \left\langle w^2 \right\rangle \left\langle \frac{U_3}{\Psi} \right\rangle  
+ 2\left(\left\langle w^2 \right\rangle \left\langle \frac{U_2}{\Psi} \right\rangle^2 - 
\left\langle w^2 \frac{U_2}{\Psi} \right\rangle \left\langle \frac{U_2}{\Psi} \right\rangle \right) \right) 
\eea 
One has, from \eqref{U2U3} 
\be
\frac{U_2}{\Psi}  =  \bar \kk \Psi,   \quad  \quad \frac{U_3}{\Psi}  =  - \frac{1}{6} \Psi^2 + \frac{1}{2}   (\bar \kk \Psi)^2 + \bar \kk (\Psi \bar \kk \Psi) 
\ee

Using the above notation and recalling that $s_1= \int_{\I \mathbb{R}} \frac{dw}{2 \I \pi} \Psi(w)$ one has
\be 
\frac{1}{s_1} \left\langle \frac{U_2}{\Psi} \right\rangle = \langle \bar {\sf k}(w_1,w_2) \rangle,  
\quad  \quad \frac{1}{s_1} \left\langle w^2 \frac{U_2}{\Psi} \right\rangle = \langle w_1^2 \bar {\sf k}(w_1,w_2) \rangle 
\ee 
Note that using \eqref{c2derivativeL} together with the symmetry in $w_1,w_2$
\be 
\partial_L \langle \bar {\sf k}(w_1,w_2) \rangle = 2 \langle (w_1^2 - \langle w^2 \rangle )\bar {\sf k}(w_1,w_2) \rangle
\ee 
and 
\be 
\frac{1}{s_1^2} \left\langle \frac{U_3}{\Psi} \right\rangle = - \frac{1}{6} \langle \delta(w_1-w_2) \delta(w_2-w_3)  \rangle  + \frac{1}{2} 
\langle \bar k(w_1,w_2) \bar \kk(w_1,w_3) \rangle + 
\langle \bar k(w_1,w_2) \bar \kk(w_2,w_3) \rangle 
\ee 
Using the identity \eqref{c2derivativeL} 
we finally obtain
\bea 
&& c_3 =   \partial_L 
\left( - \frac{1}{6} \langle \delta(w_1-w_2) \delta(w_2-w_3)  \rangle  \right) - \frac{3}{2} \partial_L  \langle \bar {\sf k}(w_1,w_2) \rangle^2 \\
&& + \frac{1}{2} 
\langle (w_1^2 - \langle w^2 \rangle )  \bar \kk(w_1,w_2) \bar \kk(w_1,w_3) \rangle + 
\langle (w_1^2 - \langle w^2 \rangle ) \bar \kk(w_1,w_2) \bar \kk(w_2,w_3) \rangle 
\eea 
The last two terms can be combined into a derivative w.r.t. $L$, using again \eqref{c2derivativeL} 
together with the permutation symmetry in $w_1,w_2,w_3$, leading to the final result
\be 
 c_3 = \partial_L \left( \frac{1}{2}  \langle \bar \kk(w_1,w_2) \bar \kk(w_1,w_3) \rangle 
- \frac{3}{2}  \langle \bar {\sf k}(w_1,w_2) \rangle^2 
- \frac{1}{6} \langle \delta(w_1-w_2) \delta(w_2-w_3)  \rangle \right) 
\ee

\subsection{Extension for $u\leq 0$ or $v\leq 0$}
\label{sec:analytic}
All previous formulas are valid for $u,v>0$. When $u\to 0$ with $v>0$ or $v\to 0$ with  $u>0$, there is actually no singularity. We obtain, for instance in the case $u = 0, v>0$, using the duplication formula of the Gamma function, 
$$ \Psi(w) = 4\cos(\pi w ) \Gamma(v+w)\Gamma(v-w) e^{w^2L}, $$
which has no singularities on the vertical line $\I\R$ nor on $\delta+\I\R$ for small $\delta>0$. 
Thus, all the formulas above can be used directly in that case. In all other cases, i.e. $u<0$ or $v<0$ or $u=v=0$, we need to perform analytic continuations of the above formulas in $u,v$.

\subsubsection{First cumulant} In order to perform the analytic continuation of $c_1(u,v,L)$, we write,  using \eqref{eq:firstorderintro},  
\be 
c_1(u,v,L) = \frac{-1}{24} + \frac{1}{2} \frac{\partial_L \widetilde{\mathcal{Z}}_{u,v}(L)}{\widetilde{\mathcal{Z}}_{u,v}(L)}
\label{eq:c1Ztilde}
\ee 
where $\widetilde{\mathcal{Z}}_{u,v}(L)=\frac{1}{\Gamma(u+v)} \int_{\I\R } \frac{dw}{2\I\pi} \Psi(w)$. The division by $\Gamma(u+v)$ is convenient to handle the case where both $u$ and $v$ are negative. The function $\Psi$ in \ref{defPsi} has poles at $- u-i, u+i, -v-i$ and  $v+i$  for all $i\in \Z_{\geq 0}$. In the phase $u,v>0$, the singularities at $u+i$ and $v+i$ lie to the right of the integration contour $\I\R$, while singularities at $-u-i$ and $-v-i$ lie to the left of $\I\R$. When $u$ or $v$ is negative, some of these singularities go to the other side of $\I\R$.  The analytic continuation of $\widetilde{\mathcal{Z}}_{u,v}(L)$ in the variable $u$ or $v$ (or both) is given by \cite[Section 4.5]{barraquand2024integral}
\begin{multline}
    \widetilde{\mathcal{Z}}_{u,v}(L) = \frac{1}{\Gamma(u+v)} \int_{\I\R } \frac{dw}{2\I\pi} \Psi(w) + \frac{1}{\Gamma(u+v)} \sum_{i=0}^{\lfloor -u\rfloor} \Res{z=-u-i}\left\lbrace \Psi(z) \right\rbrace - \Res{z=u+i}\left\lbrace \Psi(z) \right\rbrace \\ + \frac{1}{\Gamma(u+v)}\sum_{i=0}^{\lfloor -v\rfloor} \Res{z=-v-i}\left\lbrace \Psi(z) \right\rbrace - \Res{z=v+i}\left\lbrace \Psi(z) \right\rbrace. 
    \label{eq:analyticcontinuationZtilde}
\end{multline}
The formula above is correct when $u,v$ are not negative integers, and has to be modified slightly for negative integers. Figure \ref{fig:papillon} explains how the residues above arise in the analytic continuation. 
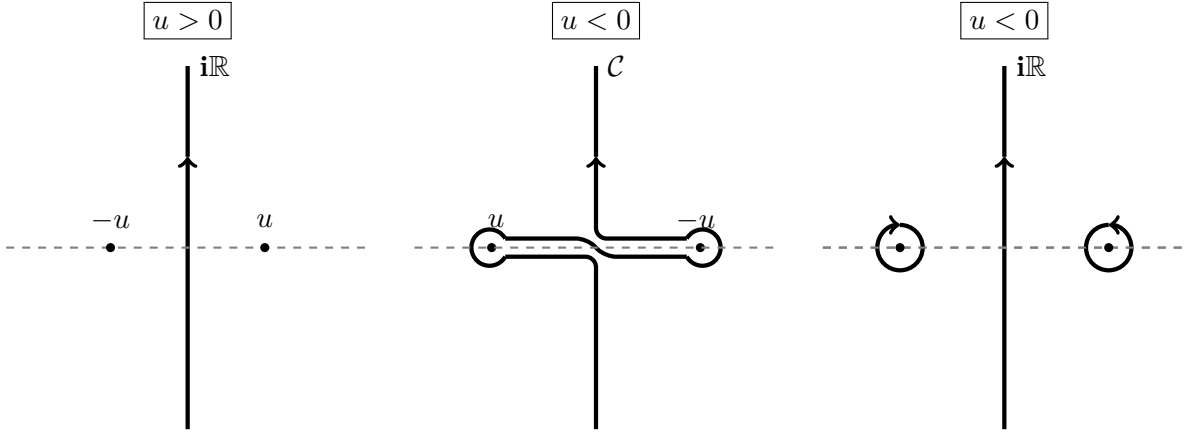
\begin{figure}
    \centering
        \begin{tikzpicture}[scale=0.6]
\draw[ultra thick, ->] (0,-4) -- (0,2);
\draw[ultra thick] (0,2) -- (0,4) node[anchor = west]{$\I\mathbb R$};
\draw[gray, thick, dashed] (-4,0) -- (4,0);
\draw (0,5) node{$\boxed{u>0}$};
\fill (1.7,0) circle(0.1);
\fill (-1.7,0) circle(0.1);
\draw (1.7,0.6) node{$u$};
\draw (-1.7,0.6) node{$-u$};
\begin{scope}[xshift=9cm]
\draw (0,5) node{$\boxed{u<0}$};
 \draw[ultra thick, rounded corners] (0,-4)  -- (0,-0.2) -- (-2,-0.2);
 \draw[ultra thick, rounded corners] (-2,0.2) -- (-0.2,0.2) -- (0.2,-0.2) -- (2,-0.2);
\draw[ultra thick, ->, rounded corners] (2,0.2)  -- (0,0.2) -- (0,2);
\draw[ultra thick] (0,2) -- (0,4) node[anchor = west]{$\mathcal C$};
\draw[ultra thick] (-2,-0.2) arc(330:30:0.4);
\draw[ultra thick] (2,-0.2) arc(-150:150:0.4);
\fill (2.3,0) circle(0.1);
\fill (-2.3,0) circle(0.1);
\draw (2.2,0.6) node{$-u$};
\draw (-2.2,0.6) node{$u$};
\draw[gray, thick, dashed] (-4,0) -- (4,0);   
\end{scope}
\begin{scope}[xshift=18cm]
\draw (0,5) node{$\boxed{u<0}$};
 \draw[ultra thick, ->] (0,-4) -- (0,2);
\draw[ultra thick] (0,2) -- (0,4) node[anchor = west]{$\I\mathbb R$};
\draw[gray, thick, dashed] (-4,0) -- (4,0);
\draw[ultra thick, ->] (-2.3,0.5) arc(90:-270:0.5);
\draw[ultra thick, <-] (2.3,0.5) arc(90:-270:0.5);
\draw[gray, thick, dashed] (-4,0) -- (4,0);   
\fill (2.3,0) circle(0.1);
\fill (-2.3,0) circle(0.1);
\end{scope}
    \end{tikzpicture}
    \caption{We consider a function $F$ defined for $u>0$ by an integral of the form $F(u) =\int_{\I\mathbb R} f(u,w)\frac{dw}{2\I\pi}$
    where $w\mapsto f(u,w)$ is a meromorphic function with poles at $w=u$ and $w=-u$, and $u\mapsto f(u,w)$ is also holomorphic except at $w=\pm u$.  
    Then, its analytic continuation to $u<0$ is $\int_{\mathcal C} f(u,w)\frac{dw}{2\I\pi}$ where $\mathcal C$ is the contour shown above. This contour is such that $u$ still lies on the right of the contour (as is the case when $u>0$), and the poles at $-u$ still lies to the left of the contour (as is the case  when $u>0$). Since the poles do not cross the contour, the formula remains analytic in $u$.  
    The contour $\mathcal C$ can be then deformed to the union of a vertical line and small circles around simple poles whose contribution can be computed by the residue theorem.}
    \label{fig:papillon}
\end{figure}
The analytic continuation of $\partial_L \widetilde{\mathcal{Z}}_{u,v}(L) $ is the same as \eqref{eq:analyticcontinuationZtilde} after replacing $\Psi(w)$ by $w^2\Psi(w)$. Computing the residues explicitly, we obtain (for $u,v\not\in\mathbb Z_{\leq 0}$)
\begin{multline}
    \widetilde{\mathcal{Z}}_{u,v}(L) = \frac{1}{\Gamma(u+v)} \int_{\I\R } \frac{dw}{2\I\pi} \Psi(w) +  \sum_{i=0}^{\lfloor -u\rfloor} \frac{2(-1)^i}{i!} \frac{\Gamma(v-i-u)\Gamma(2u+i)}{\Gamma(2(u+i))\Gamma(-2(u+i))}e^{L(u+i)^2}(u+v)_i  \\ + \sum_{i=0}^{\lfloor -v\rfloor} \frac{2(-1)^i}{i!} \frac{\Gamma(u-i-v)\Gamma(2v+i)}{\Gamma(2(v+i))\Gamma(-2(v+i))}e^{L(v+i)^2}(u+v)_i.
\end{multline}
where $(a)_k=a(a+1)\dots (a+k-1)$. For example, if $-1<u<0$ and $v>0$, we have 
\begin{equation}
     \widetilde{\mathcal{Z}}_{u,v}(L) = \frac{1}{\Gamma(u+v)} \int_{\I\R } \frac{dw}{2\I\pi} \Psi(w) + 2e^{Lu^2} \frac{\Gamma(-u+v)}{\Gamma(-2u)}.
     \label{eq:analyticZtilde}
\end{equation}
The expression for $c_1(u,v;L)$ is finally obtained by plugging the analytic continuation of $\widetilde{\mathcal{Z}}_{u,v}(L)$ in \eqref{eq:c1Ztilde}. In general, the expression is not particularly nice. However, when $u+v=0$, a drastic simplification occurs. We find that  $\widetilde{\mathcal{Z}}_{u,v}(L)  = 2 e^{L u^2}$,  
so that 
\be
c_1(u,v,L) = \frac{-1}{24}+\frac{1}{2}u^2. 
\label{eq:c1equilibrium}
\ee In this special case it turns out that the stationary height is a Brownian with drift.

More generally at large $L$ we find that 
\be 
\lim_{L \to +\infty} c_1(u,v,L) = \frac{-1}{24}+\frac{1}{2} \max( u^2\mathds{1}_{u<0} , v^2\mathds{1}_{v<0} ) 
\ee 
The second term is non-zero for either $u$ or $v$ negative, and corresponds to the binding energy per unit length
of the polymer to one of the edges of the interval.

\subsubsection{Second cumulant}
The analytic continuation of $c_2(u,v,L)$ can be performed using the same principle, i.e. starting from \eqref{eq:secondcumintro} and analytically continuing as in Figure \ref{fig:papillon}. The formulas become more complicated and we will only explain how to compute $c_2(u,v,L)$ when $u+v=0$.  According to \eqref{eq:secondcumulantbulk}, we have two terms to consider. When $u,v>0$, 
\begin{align}
    I_1&= \frac{1}{\widetilde{\mathcal{Z}}_{u,v}(L)^2}\int_{\I\R} \frac{dw_1}{2\I\pi}\int_{\I\R} \frac{dw_1}{2\I\pi} \frac{1}{\Gamma(u+v)^2} \Psi(w_1)\Psi(w_2) \kk(w_1,w_2) \left( w_1^2 - f(u,v)\right),\label{eq:I1}\\
    I_2&= \frac{1}{\widetilde{\mathcal{Z}}_{u,v}(L)^2}\int_{\I\R} \frac{dw}{2\I\pi}\frac{1}{\Gamma(u+v)^2} \Psi(w)^2\left( w_1^2 - f(u,v)\right),\label{eq:I2}
\end{align}
where $f(u,v) = \frac{\int_{\I\R} \frac{dw}{2\I\pi}\frac{1}{\Gamma(u+v)} w^2\Psi(w)}{\widetilde{\mathcal{Z}}_{u,v}(L)}$. We will study the analytic continuation of $I_1$ and $I_2$ when $u<0$, and then let $v$ go to $-u$. In order to perform this limit, we need some information on the function $f(u,v)$. We already now that $f(u,-u)=u^2$. Moreover, we will need the approximation,  for $v$ close to $-u$ and $u<0$, 
\be 
f(u,v) = u^2 + c(u+v)+ o(u+v), \text{ with }\,c= \frac{1}{2 e^{Lu^2}} \int_{\I\R} \frac{dw}{2\i\pi} \Psi(w)\Big\vert_{v=-u}(w^2-u^2).
\ee 
which is obtained from \eqref{eq:analyticZtilde}. Let us now explain how to perform the analytic continuation. In both cases, the principle is similar to \eqref{eq:analyticcontinuationZtilde}, we obtain an integral and  residues. When $u<0$ and we let $v$ go to $-u$, both integrals vanish because of the factors $\frac{1}{\Gamma(u+v)^2}$ in \eqref{eq:I1} and \eqref{eq:I2}. The residues can be computed, and after letting $v$ go to $-u$, we obtain that the residues of $I_1$ vanish, and the contribution of resiudes in $I_2$ is $-(c+u)(2e^{Lu^2})^2$. Taking into account the division by $\widetilde{\mathcal{Z}}_{u,v}(L)^2$ in \eqref{eq:I1} and \eqref{eq:I2}, we finally obtain that for $u<0$
\be
c_2(u,-u,L) = -u-\frac{1}{2 e^{Lu^2}} \int_{\I\R} \frac{dw}{2\I\pi} \Psi(w)\Big\vert_{v=-u}(w^2-u^2). 
\ee 
Moreover, we may let $u$ go to zero and recover \eqref{eq:c200}. 

\section{Large $L$ limit}

We now consider the solution of the equations \eqref{systemEmu1},\eqref{systemEmu2},
\eqref{eq:functionalequation}, \eqref{defPsi} \eqref{defoperator}, \eqref{defkernel}
in the limit $L \gg 1$. In that limit it was found in \cite{barraquand2022steady}
that the natural scaling for the boundary parameters is to set
\be 
u = \frac{\widetilde u}{\sqrt{L}},  \quad  \quad v=\frac{\widetilde v}{\sqrt{L} }.
\ee 
Here we 
will first examine the equations for $L \gg 1$.

\subsection{General formula}
For the moment we assume that $u,v>0$. From the definition of the function $\Psi_L(w)$ (making here explicit the dependence in $L$) given 
in \eqref{defPsi} we see that it takes the following scaling form in the large $L$ limit
\be 
\Psi\left(w= \frac{\widetilde w}{\sqrt{L}}\right) = L \widetilde \Psi(\widetilde w),  \quad  \quad 
\widetilde \Psi(\widetilde w) = \frac{-4 \widetilde w^2}{(\widetilde u^2 - \widetilde w^2)(\widetilde v^2 - \widetilde w^2)} e^{\widetilde w^2} 
\ee 

The coefficients $U_k(w)$ also simplify drastically in the scaling limit. Indeed the kernel
which arise in the convolutions when expressed in the scaled variables
take the form
\be 
\kk\left(\frac{\widetilde w_1}{\sqrt L},\frac{\widetilde w_2}{\sqrt L}\right) = -2 \left(  \psi\left(1 + \frac{\widetilde w_1 - \widetilde w_2}{\sqrt{L}} \right) 
+ \psi\left(1 + \frac{\widetilde w_2 - \widetilde w_1}{\sqrt{L}} \right)\right) = \kk(0,0) + O\left(\frac{1}{\sqrt{L}}\right) 
\ee 
where $\kk(0,0) = - 4 \gamma_E$. 
Hence from \eqref{U_k} we obtain the first three coefficients as
\begin{align}
U_1 &= \Psi = L \widetilde \Psi(\widetilde w),\\
U_2 &= \Psi^2 + \Psi ({\sf k} \Psi) = L^2 \left( \widetilde \Psi(\widetilde w)^2 + \frac{1}{\sqrt{L}} {\sf k}(0,0) \widetilde \Psi(\widetilde w) \int_{\I\R}\frac{d\widetilde w'}{2\I\pi} \widetilde \Psi(\widetilde w') + O\left(\frac{1}{L}\right) \right), \nn \\
 U_3 &= \frac{4}{3} \Psi^3 + 2 \Psi^2 ({\sf k} \Psi) + \frac{1}{2}  \Psi  ({\sf k} \Psi)^2 + \Psi ( {\sf k} \Psi^2 ) 
+ {\sf k}( \Psi ({\sf k} \Psi)) = L^3 \left( \frac{4}{3} \widetilde\Psi(\widetilde w)^3 + O\left(\frac{1}{\sqrt{L}}  \right) \right).\nn 
\end{align} 
Hence it is clear that one can drop all convolutions in the leading term, as was also noted in  \cite[Eq (IV.30)]{LazarescuThesis}. More
generally one has 
\be 
U_k\left(w= \frac{\widetilde w}{\sqrt L} \right) = L^k \left( \frac{2^{k-1}}{k}\left( \widetilde\Psi(\widetilde w)\right)^k + O(1/\sqrt L)\right).
\ee 
Since one can neglect all the contributions from the kernel to the function $U(w)$, the functional equation for $U(w)$ then becomes local in $w$ (in the sense that it can be solved for each $w$ independently) in the large $L$ limit, and the system of equations \eqref{systemEmu1},\eqref{systemEmu2},
\eqref{eq:functionalequation} becomes, upon rescaling
\begin{align}  \label{eqmuE}  
s &\simeq L^{-1/2}  \int_{\I \mathbb{R}} \frac{d \widetilde w}{ 2 \I \pi} \widetilde U(\widetilde w),\\ 
 E(s)+\frac{s}{24} &\simeq    \frac{L^{-3/2}}{2}  \int_{\I \mathbb{R}} \frac{d \widetilde w}{ 2 \I \pi} \widetilde w^2 \widetilde U(\widetilde w) , 
\end{align}
where 
\be 
\widetilde U(\widetilde w) = - \frac{1}{2} \log( 1 -\widetilde \zeta \widetilde \Psi(\widetilde w) )
\label{eq:functionalequationtilde}
\ee 
where $\widetilde \zeta \simeq 2 \zeta L$ must be eliminated. From this we see that there exists a limiting function 
\be 
\widetilde E(\widetilde s) = \lim_{L\to\infty} L^{3/2} \left( E(L^{-1/2}\widetilde s ) + \frac{L^{-1/2}\widetilde s}{24}\right) 
\label{eq:deftildeE}
\ee 
which admits the following parametric representation (where $\widetilde \zeta$ should be eliminated)
\be \label{parametricrescaled}
\widetilde s = \frac{1}{2} \sum_{k=1}^{+\infty} \frac{1}{k}  \widetilde\zeta^k \phi_k, \;\;\; \widetilde E(\widetilde s) = -\frac{1}{4} \sum_{k=1}^{\infty} \frac{1}{k}  \widetilde\zeta^k \psi_k
\ee 
obtained by expanding the logarithm in \eqref{eq:functionalequationtilde}, where 
\be 
\phi_k=\int_{\I\R} \frac{d\widetilde w}{2\I\pi} \widetilde \Psi(\widetilde w)^k, \;\;\; \psi_k = -\int_{\I\R} \frac{d\widetilde w}{2\I\pi} \widetilde w^2\widetilde \Psi(\widetilde w)^k.
\ee
Note that the coefficients $\phi_k$ and $\psi_k$ can also be written as real integrals 
\be \label{defphipsi} 
\phi_k=\int_{\R} \frac{dy}{2\pi}  \phi(y)^k, \;\;\; \psi_k=\int_{\R} \frac{dy}{2\pi}  y^2 \phi(y)^k
\ee 
where 
\be 
\label{defphi} 
\phi(y) = \frac{4 y^2 e^{-y^2} }{(\widetilde u^2+y^2)(\widetilde v^2+y^2)}.
\ee 
The coefficients $\phi_k$ and $\psi_k$ are bounded by $C^k$ for some $C>0$ so that the power series in $\tilde\zeta $ in \eqref{parametricrescaled} have a finite radius of convergence. 

We can also rewrite the parametric representation of $E(\widetilde s)$
as 
\be \label{parametricrescalednew}
 \widetilde s = - \frac{1}{2} \int_{\R} \frac{dy}{2\pi} \log(1 - \widetilde\zeta \phi(y)), 
 \quad  \quad 
 \widetilde E(\widetilde s) = \frac{1}{4} \int_{\R} \frac{dy}{2\pi}  y^2 \log(1 - \widetilde\zeta \phi(y)) 
\ee
where one must eliminate $\tilde \zeta$. This representation is valid
in some interval around $\tilde s=0$. 
\begin{remark}
In light of Remark \ref{rem:universalityPierre} and \ref{rem:universalityGuillaume}, the function $E(\tilde s)$ should be the large time cumulant generating function of the KPZ fixed point, and can be recovered by taking a large time large scale limit of cumulant generating function of other models in the KPZ class. In particular, we have checked that our formulas \eqref{parametricrescalednew} agrees with the corresponding formulas obtained from open TASEP, see \cite[Eq. (12)]{prolhac2024kpz} (the result was extracted in \cite{godreau2021riemann} from formulas in \cite{lazarescu2011exact}). In order to match precisely the formula \eqref{parametricrescalednew} and \cite[Eq. (12)]{prolhac2024kpz}, we set $\tilde u=\sigma_a$, $\tilde v=\sigma_b$ and $\frac{-1}{\tilde\zeta} = \frac{4 e^{-\nu}}{(1+\sigma_a^2)(1+\sigma_b^2)} $.
\end{remark}
\begin{remark}
Using that $\partial_L \Psi(w) = w^2 \Psi(w)$ 
one finds that function $\widetilde \Psi(\widetilde w)$ satisfies 
    \be 
  \left(   1   + \frac{1}{2} \left( \widetilde u \partial_{\widetilde u} + \widetilde v \partial_{\widetilde v} + 
    \widetilde w \partial_{\widetilde w} \right) \right)  \widetilde \Psi(\widetilde w) 
    = \widetilde w^2 \widetilde \Psi(\widetilde w)  
    \ee 
    Multiplying by $\widetilde \Psi(\widetilde w)^{k-1}$, integrating over $\widetilde w$ and performing an integration by part
    it implies the relation 
    \bea \label{relationphipsi} 
    \psi_k = - \left(  1 - \frac{1}{2 k}   + \frac{1}{2 k} \left( \widetilde u \partial_{\widetilde u} + \widetilde v \partial_{\widetilde v} \right) \right) \phi_k.
    \eea 
\end{remark}
\subsection{Cumulants} 
\label{sec:cumulantslargeL}
Expanding \eqref{eq:deftildeE} in power series in $\widetilde s$, we find that large time cumulants $c_k(u,v,L)$ scale as 
\begin{align}  \label{cklargetime2} 
c_1(u,v,L) &\sim_{L\to\infty} \frac{-1}{24} + \frac{1}{L} \widetilde c_1(\widetilde u,\widetilde v)\\    c_k(u,v,L) &\sim_{L\to\infty} L^{\frac{k-3}{2}} \widetilde c_k(\widetilde u, \widetilde v) \text{ for }k\geq 2.
\end{align}
where the  $\widetilde c_k$ can be computed by extracting coefficients of 
\be 
\widetilde E(\widetilde s) = \sum_{k=1}^{+\infty} \frac{\widetilde s^k}{k!} \widetilde c_k.
\ee 
The formula are simplest in terms of the variables $\hat \psi_n= \psi_n/\phi_1^n$ and 
$\hat \phi_n= \phi_n/\phi_1^n$, one obtains 
\begin{align}
\widetilde c_1 &= - \frac{\hat \psi _1}{2 } 
=  -\frac{\psi _1}{2 \phi_1 }, \\
\widetilde c_2 &=  \left(\hat \psi _1 \hat \phi _2- \hat \psi
   _2\right) = 
\frac{1}{\phi_1^3}  \left(\psi _1 \phi _2-\psi_2 \phi_1\right) \\
 \widetilde c_3 &= 6 \left(-2 \hat \psi
   _3-4 \hat \psi _1 \hat \phi _2^2+4
   \hat \psi _2 \hat \phi _2+2 \hat \psi _1
   \hat \phi
   _3\right) =
\frac{2}{\phi_1^5}  \left(-2 \psi
   _3 \phi_1^2 -3 \psi _1 \phi _2^2+3
   \psi _2 \phi _2 \phi_1 +2 \psi _1
   \phi
   _3 \phi_1 \right) \nn 
\end{align}

For the first moment it gives (for $\tilde u>0,\tilde v>0$) 
\be 
\widetilde c_1(\tilde u,\tilde v)= - \frac{1}{2}  \frac{\int_{\R}\frac{y^4}{(\widetilde u^2+y^2)(\widetilde v^2+y^2)}e^{-y^2}dy }{\int_{\R}\frac{y^2}{(\widetilde u^2+y^2)(\widetilde v^2+y^2)}e^{-y^2}dy}.
\ee  
Note that it vanishes at the point $(\tilde u,\tilde v)=0$, indeed for $0 < \tilde u \ll 1$ 
and $0< \tilde v \ll 1$ one has 
$\widetilde c_1(\tilde u,\tilde v) \simeq - \frac{\tilde u + \tilde v}{2 \sqrt{\pi} } $.
This agrees with the result \eqref{eq:c1equilibrium} which in the 
large $L$ limit implies 
$\tilde c_1(\tilde u, -\tilde u) = \frac{\tilde u^2}{2}$, from \eqref{cklargetime2}.
Some more explicit expressions and limiting behaviors for the the first and higher cumulants 
are given in Appendix \ref{app:formula}.

In some limiting cases, we can make contact with previous results. The next section discusses the maximal current phase, that is  $\widetilde u, \widetilde v=+\infty$. We discuss in Appendix \ref{app:formula} another case, the line between the maximal current phase and the high density phase, that is $\tilde u\geq 0, \tilde v=0$.

\subsection{Maximal current phase $\widetilde u, \widetilde v=+\infty$}

In this limit one has from \eqref{defphi}  
\be \phi(y) \simeq 4 y^2 e^{-y^2}/(\tilde u^2 \tilde v^2).
\ee
Using \eqref{parametricrescaled} and \eqref{defphipsi} and redefining $\hat \zeta=  \tilde \zeta/(\tilde u^2 \tilde v^2)$
we obtain the parametric representation (where $\hat \zeta$ should be eliminated): 
\be 
\widetilde s = \frac{1}{4 \sqrt{\pi}}\sum_{k=1}^{+\infty} \hat\zeta^k 
 \frac{(2 k)!}{k! k^{\frac 3 2 +k}} , \;\;\; \widetilde E(\widetilde s) = -\frac{1}{8 \sqrt{\pi}} \sum_{k=1}^{\infty} \hat \zeta^k  \frac{(2 k)!}{k! k^{\frac 5 2 +k}} \left(\frac{1}{2} + k\right).
\ee 
This can be made in correspondence with the result for the maximum current phase of the
ASEP in \cite[Eq. (IV.61-62)]{LazarescuThesis}
as follows (we identify our variable $\hat \zeta$ with the variable  $B$ in \cite{LazarescuThesis}):  
\begin{align} \label{corresp}
2 \widetilde s &= - \mu L^{1/2},  \\
 \widetilde E(\widetilde s) + \frac{\widetilde s}{2} &= \left( \frac{E}{1-q} - \frac{1}{4} \mu\right) L^{3/2} .
\end{align} 
Hence up to a scale factor and a shift in the first cumulants, the generating functions of the
cumulants are the same in this limit, which is expected by universality -- see Remark \ref{rem:universalityPierre} and Remark \ref{rem:universalityGuillaume}. Let us also recall that
in that limit,  the stationary height is the sum of a Brownian motion and a Brownian excursion \cite{derrida2004asymmetric, bryc2019limit}. We can give the corresponding cumulants
\be \label{cumulantsmaximal}
\tilde c_1= -\frac{3}{4} , \quad  \tilde c_2= \frac{3
   \sqrt{\frac{\pi }{2}}}{8}
   , \quad   \tilde c_3= 6
   \left(\frac{40}{81
   \sqrt{3}}-\frac{9}{32}\right)
   \pi 
  , \quad    \tilde c_4= \frac{5}{192}
   \left(567+486 \sqrt{2}-512
   \sqrt{6}\right) \pi^{3/2} 
\ee 
Some correction terms for large $\tilde u, \tilde v$ are given in Appendix \ref{app:formula}.

\section{Cumulants from the replica Bethe ansatz method of Brunet and Derrida}
\label{sec:BrunetDerrida}
Here we extend to the case of the interval the Bethe ansatz replica method of Brunet and Derrida \cite{brunet2000ground,brunet2000probability},
originally developed for the periodic boundary conditions. The aim is to compute the cumulants $c_k$
from 
\be  \label{E0cum} 
- E_0(n,L) = \lim_{t \to +\infty} \frac{1}{t} \log  \mathbb{E}[Z(x,t)^n] = \sum_{k \geq 1} \frac{c_k(u,v,L)}{k!} n^k.
\ee 
The quantity $E_0(n,L)$ can be computed for positive integer $n$ using the replica Bethe ansatz. 
It is equal to the ground state energy of a delta Bose gas with $n$ particles, on the interval (see below). 
We will guess an analytic continuation in $n$ to obtain the $c_k$. This program was
carried out succesfully in 
\cite{brunet2000ground,brunet2000probability}, in the case of periodic boundary conditions. We will remain close to the notations used in these works.

\subsection{Replica Bethe ansatz}

Let us define the $n$-point moments 
\be 
C_n(x_1,\dots,x_n;t)= \mathbb{E}[ Z(x_1,t) \dots Z(x_n,t) ] 
\ee 
From \eqref{eq:SHE} they obey the evolution equation 
\be
\partial_t C_n = - H_n C_n \quad , \quad H_n = - \frac{1}{2} \sum_{\alpha=1}^n \partial_\alpha^2 - \gamma \sum_{1 \leq \alpha<\beta \leq n} \delta(x_\alpha - x_\beta) 
\label{eq:defHn}
\ee 
with the boundary conditions inherited from \eqref{eq:SHEboundaryconditions}. The noise strength parameter $\gamma$ is 
set to unity in this paper, but kept for convenience here. 
The spectrum of $H_n$ on the interval is solvable by Bethe ansatz. For general references see \cite{gaudin} and
\cite{diejen2017orthogonality, diejen2018completeness}, for the hard box see \cite{oelkers2006bethe}, and the numerical study in 
\cite{chinesehardbox}. In the context of the KPZ equation see \cite[Eq. (19-24)]{denardis2020delta}. 
The symmetric eigenstates $\Psi_{\boldsymbol\lambda}(\mathbf x)$ of $H_n$, with $\mathbf x=(x_1,\dots,x_n)$,
i.e. the solutions of $(H_n \Psi_{\boldsymbol\lambda})(\mathbf x) = E_{\boldsymbol\lambda} \Psi_{\boldsymbol\lambda}(\mathbf x)$,
have the Bethe form parametrized by a set of rapidities 
$\boldsymbol \lambda= (\lambda_1 , \cdots , \lambda_n) \in \mathbb{C}^n$,
with eigenenergy 
\be 
E_{\boldsymbol\lambda} = \frac{1}{2} \sum_{\alpha=1}^n \lambda_\alpha^2. 
\ee 
Let us give the unormalized wavefunctions $\Psi_{\boldsymbol\lambda}(\mathbf x)$ for completeness. They are symmetric in the variables $\lbrace x_1, \cdots ,x_n\rbrace$. In the sector $\mathbb W_n$ defined by inequalities $0 \leqslant x_1 \leqslant  \dots   \leqslant x_n \leqslant L $, wavefunctions are defined by 
\begin{align} \label{EqT:wave}
\Psi_{\boldsymbol\lambda}(\mathbf x) &= \frac{1}{(2 \I)^{n}} \sum_{\sigma \in S_n} \sum_{\boldsymbol\epsilon \in \lbrace -1,1\rbrace^n } \epsilon_1 \epsilon_2\dots \epsilon_n {\cal A}[\epsilon_1 \lambda_{\sigma_1},\epsilon_2 \lambda_{\sigma_2}, \dots \epsilon_n \lambda_{\sigma_n}]  \exp\left( \I \sum_{\alpha=1}^n \epsilon_\alpha \lambda_{\sigma_\alpha} x_\alpha   \right) \nn  \\
 {\cal A}[\lambda_1,..,\lambda_n] &= \prod_{1\leqslant \alpha < \beta \leqslant n } 
\left(1+ \frac{\I \gamma}{\lambda_\beta - \lambda_\alpha}\right)\left(1+ \frac{\I \gamma}{\lambda_\beta + \lambda_\alpha}\right)
\prod_{\alpha=1}^n \left(1 + \I \frac{\lambda_\alpha}{A}\right).
\end{align}
where the sum is over permutations $\sigma$ of $S_n$ and signatures $\boldsymbol{\epsilon}=(\epsilon_1,\dots,\epsilon_n) $ and we denote the boundary parameters as $A=u-1/2$ and $B=v-1/2$.
Such wavefunctions are continuous on $(\mathbb{R}_{\geq 0})^n$ and satisfy the  boundary conditions at $x=0$ consistent with \eqref{eq:SHEboundaryconditions},
as well as the conditions at the boundary of the sector $\mathbb W_n$, which arise from the $\delta(x_i-x_j)$ interaction, namely:
\be
 (\partial_{x_{i+1}} -  \partial_{x_{i}} + \gamma) \Psi_\mu(x_1 \leqslant \cdots\leqslant  x_n)|_{x_{i+1}=x_i^+}=0.
\ee 
Then, enforcing the boundary condition at $x=L$ is equivalent to impose that the rapidities satisfy the following Bethe equations
\bea \label{EqT:BetheEq1}
e^{2 \I \lambda_\alpha L} = \frac{B - \I \lambda_\alpha}{B + \I \lambda_\alpha} 
\frac{A - \I \lambda_\alpha}{A + \I \lambda_\alpha} \, \prod_{1 \leqslant \beta \neq \alpha \leqslant n} \frac{\lambda_\alpha - \lambda_\beta - \I \gamma }{\lambda_\alpha - \lambda_\beta + \I \gamma} 
\frac{\lambda_\alpha + \lambda_\beta  - \I \gamma }{\lambda_\alpha  + \lambda_\beta +  \I \gamma },  \,  
 \quad \forall \alpha \in \llbracket 1,n\rrbracket \, .
\eea 
It is obvious from \eqref{EqT:wave} that changing any $\lambda_j$ to $-\lambda_j$ (for any given $j$) just changes the wavefunction  $\Psi_{\boldsymbol\lambda}(\mathbf x)$ to $- \Psi_{\boldsymbol\lambda}(\mathbf x)$. 
The eigenstate thus remains the same. 
This symmetry must be taken into account when labeling and classifying the states.

\begin{remark}
In the limit $\gamma \to 0$ (no interaction) the wavefunction \eqref{EqT:wave} becomes:
\be
 \Psi_{\boldsymbol\lambda}(\mathbf x) = \rm Perm \left( \sin (\lambda_i x_j) + \frac{\lambda_i}{A} \cos(\lambda_i x_j) \right)_{i,j=1}^n
\ee
where ${\rm Perm}(M)=\sum_{\sigma\in S_n} \prod_{i=1,n} M_{i,\sigma_i}$ is the permanent of the $n \times n$ matrix $M$, 
as expected for a non-interacting bosonic eigenstate. Inside the permanent, the function $x\mapsto \sin(\lambda x) +\frac{\lambda}{A}\cos(\lambda x)$ is a parametrization of solutions for $n=1$. For general $n$, the $\lambda_i$  take
value in the set of the solutions of the equation
\be \label{1}
e^{2 i \lambda} = \frac{A -i \lambda}{A + i \lambda } \frac{B- i \lambda }{B+i \lambda}.
\ee 
For $\gamma=0$ some of the $\lambda_i$ may coincide. The non-interacting ground state
corresponds to all $\lambda_i=\lambda_{\min}$ where $\lambda_{\min}$ is the
solution of \eqref{1} with smallest value of $\lambda^2$ (note that $\lambda^2$ is real, i.e. $\lambda$ is real or imaginary). 
\end{remark}
\begin{remark}
Back to the case of $\gamma>0$, there are many branches of solutions of the Bethe equations \eqref{EqT:BetheEq1}, which as $\gamma \to 0$ 
converge to the solutions of the non-interacting case. However, for generic values of $\gamma>0$,  all the $\lambda_i$ are distinct.
Some of them may thus tend to each others as $\gamma \to 0$.
\end{remark}

\subsubsection{A functional of the Bethe roots} 
To make easier the comparison with \cite{brunet2000ground,brunet2000probability} we will introduce the 
notations used there
\be 
\I \lambda_\alpha = \frac{2}{L} q_\alpha, \quad  \quad c = \gamma L/2, 
\ee 
and from now on set $\gamma=1$. In these variables we want to compute 
\be \label{E0} 
E_0(n,L) = \frac{1}{2} \sum_{\alpha=1}^n \lambda_\alpha^2 = - \frac{2}{L^2} \sum_{\alpha=1}^n q_\alpha^2 
\ee 
for the solution of the Bethe equations which corresponds to the ground state. The $\{ q_\alpha \}$ satisfy 
the Bethe equations (note the exponent  $4$ in the left hand side) 
\be \label{Betheab} 
e^{4 q_\alpha} = \prod_{\beta \neq \alpha} \frac{q_\alpha - q_\beta + c}{q_\alpha - q_\beta - c}
 \frac{q_\alpha + q_\beta + c}{q_\alpha + q_\beta - c} \frac{a-q_\alpha}{a+q_\alpha} \frac{b-q_\alpha}{b+q_\alpha} 
\ee 
where we have introduced the parameters $a=A L/2$, $b=B L/2$. 

Let us define an index set  $I_n=\lbrace 1,\dots, n, -1, \dots, -n\rbrace$ and a set of complex numbers  $\mathcal Q=\lbrace q_{\alpha}\rbrace_{\alpha\in I_n} = \lbrace q_1, \dots, q_n, q_{-1}, \dots, q_{-n} \rbrace$ where we use the notational convention  $q_{-i}=-q_i$ and  for $\alpha \in I_n $, 
$$ \rho_\alpha = \prod_{q\in \mathcal Q\setminus q_{\alpha}} \frac{q_{\alpha}-q+c}{q_\alpha-q}.$$
Next, we define a function $L(z)$ as 
\be  \label{defLz} 
L(z) = \sum_{\alpha\in I_n}  \rho_\alpha e^{q_{\alpha}z}.
\ee 
It satisfies
\be 
L(0)= 2 n, \quad  \quad L'(0)= n(2n-1) c, \quad  \quad L''(0) = \frac{2}{3} n (n-1) (2 n-1) c^2 + 2 \sum_{\alpha=1}^n q_\alpha^2.
\ee
The last three equations comes from the relations 
\begin{align}
    &\sum_{\alpha \in I_n} \rho_\alpha=2n,\\
    &\sum_{\alpha \in I_n} q_{\alpha}\rho_\alpha=c\binom{2n}{2},\\
    &\sum_{\alpha \in I_n} q_{\alpha}^2 \rho_\alpha=c^2 \binom{2n}{3} + \sum_{\alpha=-n}^n (q_{\alpha})^2,
\end{align}
which are proved in \cite[Appendix A, Eq. (A6-A8)]{brunet2000probability}, where here the set $I_n$ has $2 n$ elements instead of $n$. 

\subsubsection{Integral equation and Bethe equations satisfied by the functional}
We now show that the function $L(z)$ satisfies the integral equation 
\be \label{eqL}
    \int_0^w e^{-c z/2} L(-z) e^{c (w-z) /2} L(w-z)dz = 
     \frac{1}{c} \left(L(w) e^{c w/2} - L(-w) e^{-c w/2}\right)  
\ee 
coming from the following identities, valid for any $\beta\in I_n$: 
\be \sum_{\alpha \in I_n} \frac{\rho_\alpha}{q_{\alpha}-q_{\beta}+c} = \sum_{\alpha \in I_n} \frac{\rho_\alpha}{q_{\alpha}+q_{\beta}+c} =\frac{1}{c}
\ee 
which are proved in \cite[Appendix A, Eq. (A9)]{brunet2000probability}.
Indeed, one has
\begin{align} 
    \int_0^w e^{-c z/2} L(-z) e^{c (w-z) /2} L(w-z)dz &= 
    \sum_{\alpha \in I_n }  \sum_{\beta \in I_n } \frac{\rho_\alpha \rho_\beta}{q_\alpha + q_\beta + c} \left( e^{q_\beta w + \frac{c w}{2}} - 
    e^{- q_\alpha w - \frac{c w}{2}} \right) \\
 &= \frac{1}{c} \left(L(w) e^{c w/2} - L(-w) e^{-c w/2}\right) 
\end{align}

Now one can check that the Bethe equations \eqref{Betheab} can be put under the following form: for all $\alpha\in \lbrace 1, \dots, n\rbrace$, 
\be 
e^{2 q_\alpha} ( 2 q_\alpha+ c ) (a+q_\alpha)(b+q_\alpha)   \rho_{-\alpha} = e^{- 2 q_\alpha} (2 q_\alpha- c)   (a-q_\alpha)(b-q_\alpha) \rho_\alpha  .
\label{eq:Betheequationsrhotilde22}
\ee 
Indeed, by the definition of $\rho_{\alpha}$, 
\begin{multline} 
\frac{\rho_{\alpha}}{\rho_{-\alpha}} = \frac{2q_\alpha+c}{2q_\alpha} \frac{-2 q_\alpha}{-2q_\alpha +c} \prod_{1\leq \beta\leq n, \beta\neq \alpha} \frac{q_\alpha - q_\beta +c}{-q_\alpha - q_\beta +c}\frac{q_\alpha + q_\beta +c}{-q_\alpha + q_\beta +c}\\  =  \frac{2q_\alpha+c}{2q_\alpha-c }  \prod_{1\leq \beta\leq n, \beta\neq \alpha} \frac{q_\alpha - q_\beta +c}{q_\alpha - q_\beta -c}\frac{q_\alpha + q_\beta +c}{q_\alpha + q_\beta -c},
\end{multline}
so that \eqref{Betheab} can be rewritten as \eqref{eq:Betheequationsrhotilde22}. Note that \eqref{eq:Betheequationsrhotilde22} is stil true for $\alpha\in \lbrace -1, \dots, -n\rbrace$ as can be checked by changing $q_\alpha$ in $-q_{-\alpha}$.  From the definition of
$L(z)$ in \eqref{defLz}, we see that the symmetry condition \eqref{eq:Betheequationsrhotilde22} 
implies the following equation 
\be \label{sym2}
(2 \partial_z - c) (a-\partial_z) (b-\partial_z) L(z) = 
- (2 \partial_y - c) (a-\partial_y) (b-\partial_y) L(y)|_{y=-(4+z)}
\ee 
which relates derivatives of $L$ at argument $z$ and $-(4+z)$. We will refer to this equation as the functional Bethe equation. 

\subsubsection{A more convenient functional}
The form of \eqref{eqL} suggests to define the function 
\be \label{defM} 
M(z) = \frac{1}{2 n} e^{cz/2}L(z) 
\ee 
which satisfies 
\be 
M(0)= 1, \quad \quad M'(0) = n c, \quad  \quad M''(0) = \frac{1}{12} c^2 (1 + 8 n^2) + \frac{1}{n} \sum_{\alpha=1}^n q_\alpha^2 ,
\ee 
Hence, from \eqref{E0} we see that if one knows $M(z)$ one can extract the energy as
\be  \label{E0nL}
E_0(n,L) = - \frac{2 n}{L^2} \left( M''(0) - \frac{1}{12} c^2 (1 + 8 n^2) \right) 
\ee 
Using \eqref{defM} and \eqref{eqL} we find that $M(z)$ satisfies the integral equation
\be 
2 n c \int_0^z M(-z') M(z-z')dz' =  M(z) - M(-z). 
\label{integralequationM}
\ee 
The functional Bethe equation \eqref{sym2} now becomes 
\begin{multline} \label{MBethe} 
 e^{-c z/2} (\partial_z-c) (a + c/2 - \partial_z) (b + c/2 - \partial_z) M(z) \\  
= -  e^{-c y/2} (\partial_y-c) (a + c/2 - \partial_y) (b + c/2 - \partial_y) M(y)|_{y=-(4+z) }.
\end{multline}
We recover the parameters $u,v$ defined in \eqref{eq:boundary} appearing in the equation above since  $cu=a+c/2$, $cv= b+c/2$. 

\begin{remark}
If we define a function $B$ via 
 \be \label{MtoB}
 M(z) = e^{\frac{-c z^2}{4}}B(1+z), 
\ee 
this function $B$ is analogous to the function $B$ introduced in \cite{brunet2000ground, brunet2000probability} to
    study the periodic case in the sense that the integral equation \eqref{integralequationM} would be the same as in \cite{brunet2000ground, brunet2000probability}. 
In the periodic case, the functional Bethe equation takes the simple form $B(x)=B(-x)$. 
In the present case of the interval the functional Bethe equation \eqref{MBethe} is clearly more involved.
\end{remark}

\subsection{Expansion of $M(z)$ in powers series in $n$: first approach} \label{subsec:expansion1}
We write 
\be 
M(z) = \sum_{k=0}^{\infty} n^k m_k(z).
\ee 
We will determine the coefficients $m_k(z)$. The integral equation \eqref{integralequationM} allows to determine the coefficients by recurrence knowing the previous ones (and using the functional Bethe equation). 
In order to deal with the functional Bethe equation, it is convenient to write 
\be 
M(z) = \int_{\I\R} \frac{dw}{2\I\pi} \hat M(w) e^{czw}, \;\;\;\; m_k(z) = \int_{\I\R} \frac{dw}{2\I\pi} \hat m_k(w) e^{czw}.
\ee 

The functional Bethe equation \eqref{MBethe} may now be written as 
\begin{multline} 
\int_{\I\R} \frac{dw}{2\I\pi}  \hat M(w)(w-1)(u-w)(v-w) e^{cz(w-1/2)} \\  = -\int_{\I\R}  \frac{dw'}{2\I\pi}  \hat M(w')(w'-1)(u-w')(v-w') e^{-c(4+z)(w'-1/2)}.
\label{eq:FBEhat}
\end{multline}
Let us use the change of variables $w'=1-w$ in the right hand side of \eqref{eq:FBEhat}, and deform back the contour for $w'$ to the imaginary axis $\I\R$ -- here, we implicitly assume that the function $w\mapsto \hat M(w)(w-1)(u-w)(v-w)$ is holomorphic in the strip 
\be \strip=\lbrace w\in \C; 0\leq \Re[w]\leq 1\rbrace,
\label{eq:defstrip}
\ee 
so that no singularities are encountered during the deformation of the contour. We see that the equation \eqref{eq:FBEhat} will be satisfied if one has 
\be 
(w-1)(u-w)(v-w)\hat M(w) = w(u-1+w)(v-1+w)\hat M(1-w)e^{-2c + 4cw}.
\label{eq:FBEhat2}
\ee 
Since the functional Bethe equation does not depend explicitly on $n$, each coefficient $\hat m_k(w)$ will likewise be required to  satisfy   \eqref{eq:FBEhat2} and $\hat m_k(w)(w-1)(u-w)(v-w)$ will be required to 
be holomorphic in the strip $\strip$.

Let us consider the first order coefficient $m_0(z)$. The integral equation \eqref{integralequationM} reduces to 
\be m_0(z) = m_0(-z)\ee which is satisfied when $\hat m_0(w)=\hat m_0(-w)$. Then we may write  
\be 
\hat m_0(w) = \hat m_0(w-1) \frac{w}{w-1} \frac{u+w-1}{u-w} \frac{v+w-1}{v-w} \frac{e^{2c w^2}}{e^{2c(1-w)^2}}.
\ee 
This is reminiscent of the functional equation satisfied by the Gamma function. To exploit this, it will be convenient to write  
\be 
\hat m_0(w) = \Gamma(1+ w) \Gamma(1- w) \Gamma(u+ w) \Gamma(u- w) \Gamma(v+ w) \Gamma(v- w)  e^{2cw^2} \hat g(w) 
\ee 
and the constraints on $\hat m_0$ become  $\hat g(w) = \hat g(-w)$ and $\hat g(w)=-\hat g(1-w) $. When $u,v>0$, the condition that $\hat m_0(w)(w-1)(u-w)(v-w)$ is holomomorphic in $\strip$  corresponds to imposing that  $\hat g$ is holomorphic in $\strip$. Let us restrict to this case $u,v>0$ from now on.
There are many functions $\hat g$  satisfying these conditions. If we look for a function $\hat g$ with the smallest possible growth at infinity, one is led to the choice $\hat g(w) \propto \sin(2\pi w)\sin(\pi w)$ which corresponds to 
\be 
\hat m_0(w) = \frac{\Psi(w)}{\int_{\I\R }\frac{dw'}{2\I\pi} \Psi(w') }
\ee 
where the denominator is simply a normalizing constant chosen so that $m_0(0)=1$, and
$\Psi(w)$ is the function defined in \eqref{defPsi}. It is quite encouraging to see
that this function again emerges, within this completely different method. 

\begin{remark}
    There was some freedom in the choice of $\hat g$. A similar property was noticed 
    in the periodic case in \cite{brunet2000probability}. One can surmise
    that the other choices may correspond to other solutions of the Bethe equations,
    e.g. corresponding to excited states. 
\end{remark}

\subsection{Systematic approach} \label{subsec:systematic}
More generally, the function $\Psi(w)$ will be convenient to deal with the Bethe equation at all orders in the expansion in $n$. 
Recall that we have defined in \eqref{eq:defnu}
\be 
\nu(dw) = \frac{\Psi(w) \frac{dw}{2\I\pi}}{\int_{\I\R }\frac{dw'}{2\I\pi} \Psi(w') } 
\ee 
so that we are looking for a function $f(w)$ such that 
\be \label{decomposeM} 
M(z) = \int_{\I\R} \nu(dw) f(w)  e^{c z w},   \quad  \quad \int_{\I\R} \nu(dw) f(w) =1  
\ee 
where $M'(0)=nc$.
As we have already discussed, in order to satisfy the Bethe equations \eqref{MBethe}, the function $f(w)$ must be holomorphic in the strip $\strip$ and satisfy the symmetry $f(w) = f(1-w)$. 
Furthermore, inserting \eqref{decomposeM} into \eqref{integralequationM} we see that 
$f(w)$ must also satisfy the integral equation 
    \begin{equation} \label{eqfrestart} 
    \int_{\I\R} \nu(dw) f(w) (e^{cwz}-e^{-cwz}) = 2n \iint_{(\I\R)^2} \nu(dw_1)\nu(dw_2) f(w_1) f(w_2) \frac{e^{czw_2}-e^{-cz w_1}}{w_1+w_2}.
\end{equation}
The following lemma gives a sufficient condition on $f(w)$ to satisfy \eqref{eqfrestart}.
\begin{lemma} Let $\delta\in (0,1)$. Assume that a function $f$, holomorphic in the strip $\strip$,  satisfies for all $w_1\in \strip$, 
\be
f(w_1) = 2 n  \int_{\I\R + \delta} \nu(d w_2) \frac{f(w_1)  f(w_2)}{w_1+w_2} + h(w_1)
 \label{eqff} 
\ee 
where  $h$ is holomorphic in $\strip$ and for all $w_1\in \I\R$, $h(-w_1) = h(w_1)$. Then f satisfies \eqref{eqfrestart}.
\end{lemma} 
\begin{proof}
Plugging \eqref{eqff} in the left-hand-side of \eqref{eqfrestart}, we obtain 
\be 
2n \int_{\I\R} \nu(dw_1) \int_{\I\R+\delta} \nu(dw_2)\frac{ f(w_1) f(w_2)}{w_1+w_2} (e^{czw_1}-e^{-cz w_1}) + \int_{\I\R} \nu(dw_1) h(w_1) (e^{czw_1}-e^{-cz w_1}).
\ee 
The second integral vanishes because the integrand is antisymmetric and the contour is $\I\R$. In the first integral, we may shift the $w_1$ contour to $\I\R+\delta$  
(using the fact that $f$ is holomorphic in the strip). We obtain 
\be 2n \int_{\I\R+\delta} \nu(dw_1) \int_{\I\R+\delta} \nu(dw_2)\frac{ f(w_1) f(w_2)}{w_1+w_2} (e^{czw_1}-e^{-cz w_1}).  \ee
At this point, the factor $(e^{czw_1}-e^{-cz w_1})$ yields a difference of two integrals. In the first term, we use the change of variables $(w_1,w_2)\to (w_2,w_1)$, and we keep the second term unchanged. This yields 
\be \label{xx} 
2n \int_{\I\R+\delta} \nu(dw_1) \int_{\I\R+\delta} \nu(dw_2)\frac{ f(w_1) f(w_2)}{w_1+w_2} (e^{czw_2}-e^{-cz w_1}). 
\ee
Finally, we may shift the contours back to $\I\R$ since there is no singularity when $w_1+w_2=0$, and we arrive at the right-hand-side of \eqref{eqfrestart}. 
\end{proof}

Hence we can trade the original functional equation \eqref{eqfrestart} for the equation \eqref{eqff}, which we now study perturbatively in $n$. We assume that we can expand 
\be \label{expn} 
 f(w) = 1 + n f_1(w) + n^2 f_2(w) + \dots, \quad  \quad h(w) = 1 + n h_1(w) + n^2 h_2(w) + \dots
\ee 
Since $\int_{\I\R } d\nu(w) f(w)=1$ and $\int_{\I\R } d\nu(w)=1$, we must have $\int_{\I\R} d\nu(w) f_k(w)=0$ for $k \geq 1$. 
To order $n^0$ one has 
\be 
m_0(z) = \int_{\I\R} \nu(dw) e^{c z w},  \quad  \quad m''_0(0) = c^2 \int_{\I\R} \nu(dw) w^2 
\ee 
which allows to obtain the first cumulant $c_1=c_1(u,v,L)$ (recalling that $L=2 c$)
\be \label{firstmom}
c_1= - \lim_{n \to 0} \frac{E_0(n,L)}{n} = \frac{2}{L^2} \left(m_0''(0) - \frac{c^2}{12}\right) 
= \frac{1}{2} \left( \frac{-1}{12} +\frac{\int_{\I\R} \frac{dw}{2\I \pi} \Psi(w) w^2 }{\int_{\I\R} \frac{dw}{2\I \pi} \Psi(w) }  \right) 
\ee 
thus
recovering the result \eqref{eq:firstorderintro} obtained by the other method. 

\subsection{Expansion at order $O(n)$ and second cumulant $c_2(u,v,L)$}
At order $n$,  one must solve for $w_1 \in \strip$ (where $\strip$  is the strip defined in \eqref{eq:defstrip})
\be \label{eqf10}
f_1(w_1) = 2   \int_{\I\R + \delta} \nu(d w_2) \frac{1}{w_1+w_2}  + h_1(w_1), \quad  \quad h_1(-w_1) = h_1(w_1) 
\ee 
together with the functional Bethe equation $f(w_1)=f(1-w_1)$.
We show below that a solution of this problem is 
\be 
f_1(w_1) = 2\int_{\I\R+\delta} \nu(dw_2) K(w_1,w_2) - C_1
\label{eq:deff1}
\ee 
where $C_1$ is a constant (to be determined later) such that $\int_{\I\R} d\nu(w_1) f_1(w_1)=0$, 
and 
where the kernel $K$ is defined  by 
\be 
\label{eq:defK}
K(w_1,w_2) = \frac{1}{w_1+w_2} -\psi(1+w_1-w_2) -\psi(1+w_2-w_1).
\ee 
Indeed, we see on this form that the first term corresponds to the first term in \eqref{eqf10}, 
and we will show below that the second term leads to an even function $h_1(w_1)$
and allows $f_1(w_1)$ to satisfy the functional Bethe equation.

First of all, given the location of poles of the digamma function on $\mathbb Z_{\leq 0}$, it is clear that the function $f_1$ is holomorphic in $\strip$.  Let us now check the functional Bethe equation and the condition \eqref{eqf10}. 
In the interior of $\strip$, that is when $0 < \Re[w_1]<1$, we may shift the contour in \eqref{eq:deff1} so that  
\be 
f_1(w_1) =2  \int_{\I\R}  \nu(dw_2)  K(w_1,w_2) - C_1.
\label{eq:insidethestrip}
\ee
If we let  $w_1\to 1$ or $w_1\to 0$, the integrand becomes singular, but it  can be analytically continued: the  analytic continuation of \eqref{eq:insidethestrip} to all $w_1 \in \strip$ is precisely \eqref{eq:deff1}. 
In the interior of $\strip$, the function $f_1$ may also be rewritten as 
\be 
f_1(w_1) = \int_{\I\R}  \nu(dw_2)  \left( K(w_1,w_2)+K(w_1,-w_2)\right) - C_1.
\ee
Thanks to properties of the digamma function (see below in Remark \ref{rem:digamma}) the symmetrized kernel $K^{\rm Sym}(w_1, w_2):=K(w_1,w_2)+K(w_1,-w_2)$ satisfies $K^{\rm sym}(w_1, w_2)=K^{\rm sym}(1-w_1,w_2)$, so that $f_1(w_1)$ satisfies the Bethe equation in the interior of the strip. Moreover, if the function $f_1(w_1)-f_1(1-w_1)$ is zero inside the strip, its analytic continuation must be zero as well on $\strip$, which proves that \eqref{eq:deff1} satisfies the Bethe equation. 

Now, let us explain why $f_1$  satisfies \eqref{eqf10}. Let 
\be 
h_1(w_1):= f_1(w_1) -  2   \int_{\I\R + \delta} \nu(dw_2) \frac{1}{w_1+w_2}  = -C_1 -2 \int_{\I\R+\delta} \nu(dw_2) (\psi(1+w_2-w_1)+\psi(1+w_1-w_2))
\ee 
We need to prove that $h_1(w_1)$ is holomorphic on the strip $0\leq \Re[w_1] \leq 1$ and that for any $w_1\in \I\R$, $h_1(w_1)=h(-w_1)$. The function $z\mapsto \psi(1+z)$ has singularities only at the points $z=-1, -2, \dots$ For $0\leq \Re[w_1] \leq 1$, we have that both $\Re[w_1-w_2]\geq -\delta$ and $\Re[w_2-w_1] \geqslant  1-\delta$  which remains away from the singularities, so that $h_1$  is holomorphic as desired. Further, for $w_1\in \I\R$, we may write 
\be 
h_1(w_1) = -C_1 -2 \int_{\I\R} \nu(dw_2) (\psi(1+w_2-w_1)+\psi(1+w_1-w_2))
\ee 
so that 
      
\be 
h_1(-w_1)-h_1(w_1) = 2\int_{\I\R} \nu(dw_2) (\psi(1+w_2-w_1)+\psi(1+w_1-w_2)- \psi(1+w_2+w_1)+\psi(1-w_1-w_2))
\ee 
Since the function $w_2\mapsto \psi(1+w_2\pm w_1) - \psi(1-w_2\pm w_1)$ is antisymmetric and we integrate over $\I\R$, the integral equals $0$ so that $h$ is an even function as desired.

\begin{remark}
\label{rem:digamma}
We have postulated in \eqref{eq:defK} a very specific form for the kernel $K$. Let us explain where it comes from. First of all, in view of \eqref{eqf10}, it is natural to look for a kernel $K$ of the form 
\be 
K(w_1,w_2) = \frac{1}{w_1+w_2}+ F(w_1-w_2)+F(w_2-w_1). 
\ee 
As we have seen, under some hypotheses on $F$, a sufficient condition in order for $f_1$ to satisfy the functional Bethe equation is that the kernel $K^{\rm Sym}(w_1, w_2):=K(w_1,w_2)+K(w_1,-w_2)$ satisfies the functional Bethe equation for all $w_2$. This is equivalent to the functional equation on $F$: 
\be  
\frac{1}{-1+x} +\frac{1}{x} -F(1-x) -F(-1+x) + F(-x)+F(x)=0.
\ee
This functional equation characterizes $F(x)+F(-x)$ up to an even periodic function. To solve it, it is convenient to observe that the equation is satisfied in particular when one of the following simpler functional equations is satisfied:   
\be 
F(x+1)=\frac{-1}{x}+F(x) \text{ or }F(x-1)=\frac{1}{-1+x}+ F(x).
\ee 
We recognize the functional equation satisfied by digamma functions, 
\be 
F(x) = -\psi(1-x) \text{ or } F(x)=-\psi(1+x)
\ee 
leading to  \eqref{eq:defK}. 
\end{remark}

\subsubsection{Second cumulant}
Recall that 
\be  
E_0(n,L) = - \frac{2 n}{L^2} \left( M''(0) - \frac{1}{12} c^2 (1 + 8 n^2) \right).
\ee 
We now use the solution for $f_1(w)$ to obtain the second cumulant $c_2(u,v,L)$.
From the definition \eqref{E0cum}, and using \eqref{E0nL}
one has 
\be 
 c_2(u,v,L) = - \partial_n^2 E_0(n,L)|_{n=0}  =  \frac{4}{L^2} m_1''(0) 
\ee
where 
\begin{align*}
 m_1(z) &= \int_{\I\R} \nu(dw_1) f_1(w_1) \\
 &=2\int_{\I\R} \nu(dw_1) \int_{\delta+\I\R } \nu(dw_2) K(w_1,w_2) e^{cw_1z} - C_1 \int_{\I\R} \nu(dw_1)e^{cw_1z}.
\end{align*}  
The constraint $m_1(0)=0$ implies that 
\be 
C_1  =  2\int_{\I\R} \nu(dw_1) \int_{\delta+\I\R } \nu(dw_2) K(w_1,w_2).
\ee 
Thus, 
\be 
m_1''(0) = 2\int_{\I\R} \nu(dw_1) \int_{\delta+\I\R } \nu(dw_2) K(w_1,w_2) \left(c^2 w_1^2 - m_0''(0) \right), 
\ee 
leading to the main result of section, i.e. the following expression for the second cumulant 
 \be  
 \label{eq:c2new}
c_2(u,v,L)  =2 \int_{\I\R} \nu(dw_1) \int_{\delta+\I\R } \nu(dw_2) K(w_1,w_2) \left( w_1^2 -\int_{\I\R } d\nu(w) w^2 \right)
 \ee 
 where the kernel $K$ is defined in \eqref{eq:defK}. 
\\

We now show that the expression \eqref{eq:c2new} is identical to the result obtained by the other method in \eqref{c2exact}.
Indeed, since we can shift the contour of integration for $w_2$ in \eqref{c2exact} from $\I\R$ to $\I\R+\delta$,  the  difference between \eqref{eq:c2new} and \eqref{c2exact} is 
\be 
\int_{\I\R} \nu(dw_1) \left(  \int_{\delta+\I\R } \nu(dw_2) \frac{2}{w_1+w_2} -\widetilde \nu(w_1) \right) \left(w_1^2 - \int_{\I\R } \nu(dw) w^2 \right)
\ee 
where we defined $\widetilde \nu(w)=\Psi(w)/\int_{\I \mathbb{R}} \frac{dw'}{2 \I \pi} \Psi(w')$. To show that this difference is zero, it suffices to show that the function 
$$ \Delta(w_1):=  \int_{\delta+\I\R } d\nu(w_2) \frac{2}{w_1+w_2} -\nu(w_1) $$
is odd on $\I\R$. This is easily verified: 
\begin{align}
 \frac{1}{2} ( \Delta (w_1)+ \Delta(-w_1) ) 
&= \int_{\I\R+\delta}  \nu(dw_2) \frac{1}{w_1+w_2}  + \int_{\I\R+\delta}  \nu(d w_2) \frac{1}{-w_1+w_2}  - \nu(w_1),  \\
&= \int_{\I\R+\delta}  \nu(d w_2) \frac{1}{w_1+w_2}  - \int_{\I\R-\delta}  \nu(dw_2) \frac{1}{w_1+w_2}  - \nu(w_1), \\
&= 0, 
\end{align}
where we used the residue theorem and the fact that $\widetilde \nu(w_1)$ is an even function of $w_1$. 

\section{Periodic KPZ equation}
\label{sec:periodic}
In this section,  we sketch how the calculations of the previous sections can be repeated in
the case of the periodic case, i.e. for the KPZ equation \eqref{eq:KPZ} on a ring $x \in [0,L]$ with the condition $h(0,t)=h(L,t)$.
For this case, the calculation of the cumulant generating function,  $E_{\rm per}(s)$,  and of the cumulants
$c_k^{\rm per}(L)$ was considered 
in \cite{brunet2000probability,brunet2000ground} and in \cite{prolhac2010tree}.
Here we obtain a priori distinct formulas,  which we believe should be equivalent (we show it only in a few cases). 

Let us start with the first method, i.e. from the limit from the periodic ASEP. Using 
Remark \ref{rem:periodic} and the same rescaling as in Section \ref{sec:ASEPtoKPZ},
which leads from the ASEP
to the KPZ equation, we obtain the same equations \eqref{systemEmu1}, \eqref{systemEmu2} 
for $E_{\rm per}(s)$, where $U(w)$ is replaced by $U_{\rm per}(w)$, solution of the functional
equation
    \be 
U_{\rm per}(w) = -  \log\left( 1-2 \zeta \Psi(w) e^{\frac{1}{2} \kk U_{\rm per}(w)} \right),
\label{eq:functionalequationper}
\ee 
and $\Psi(w) = e^{2 L w^2}$. This is obtained using \eqref{functasepperiodic} and that when $N=\ell/2$, under the scalings of Section \ref{sec:ASEPtoKPZ} with $z=e^{\eps w}$,   $F_{\rm per}(z)\simeq 2^{4L\eps^{-2}} e^{\frac{L}{2}w^2}$.

From \eqref{eq:functionalequationper} we see that $U_{\rm per}(w)=\frac{1}{2} U(w)|_{\Psi(w) = e^{\frac{L}{2} w^2}}$,
where $U(w)|_{\Psi(w) = e^{\frac{L}{2} w^2}}$ is the solution of the functional equation
\eqref{eq:functionalequation}, where $\Psi(w)$ is replaced by $e^{\frac{L}{2} w^2}$. 
Hence we obtain the CGF in the periodic case as
\begin{align} 
    s &= 2 \int_{\I\R} \frac{dw}{2\I\pi} U(w),\\
    E_{\rm per}(s) &= - \frac{1}{24} s + \int_{\I\R} \frac{dw}{2\I\pi} w^2 U(w),  \label{systemEmu2per}
\end{align}
where $U(w)= U(w)|_{\Psi(w) = e^{\frac{L}{2} w^2}}$. Hence the cumulants in the periodic
case are related to those obtained from our previous calculation by simply
replacing the function $\Psi(w)$ by $e^{\frac{L}{2} w^2}$, as 
\be \label{id:per}
c_k^{\rm per}(L) = \frac{1}{2^{k-1}} \times c_k\Big|_{\Psi(w) = e^{\frac{L}{2} w^2}}
\ee 
Let us give explicitly the first two cumulants, which we easily obtain from 
\eqref{eq:firstorderintro} and 
\eqref{eq:secondcumintro}
and the identity \eqref{id:per}, which leads to 
\be \label{nuperiodic}
\nu(dw) =  \sqrt{2 \pi L} e^{\frac{L}{2} w^2} \frac{dw}{2 \I \pi} 
\ee 
The first cumulant reads
\be \label{c1per} 
c_1^{\rm per}(L) = - \frac{1}{24} - \frac{1}{2 L} 
\ee 
in agreement with \cite[Eq. (49)]{brunet2000probability}. 
The second cumulant is
\be \label{c2exactper} 
c_2^{\rm per}(L) =  \frac{1}{2} 
\int_{\I\R} \nu(dw_1) \int_{\I\R} \nu(dw_2) ( \kk(w_1,w_2) +  \delta(w_1-w_2) ) (w_1^2 - \int_{\I\R } \nu(dw) w^2 )
\ee 
which can be rewritten as 
\be \label{c2per2}
c_2^{\rm per}(L) =  \frac{1}{2}   \partial_L \left( \int_{\I\R} \nu(dw_1) \int_{\I\R} \nu(dw_2) \bar \kk(w_1,w_2) \right).
\ee 
In the Appendix \ref{app:comparisonsecondcum}
we show that this formula is equivalent to the one obtained
in \cite[Eq. (49)]{brunet2000probability}. It requires some
calculation, which becomes more tedious beyond the second cumulant. 

Finally, we have also performed the calculation using the second method,
which is sketched in Appendix \ref{app:comparisongeneral}.
Our calculation deviates from the one in \cite{brunet2000probability}
since we use the kernel $\kk$. The formulas, although they should be
equivalent, take a different form. We show their equivalence
only up to the order required to compute the
second cumulant.

\subsection*{Acknowledgments}  G. Barraquand was supported by ANR grants ANR-21-CE40-0019 and ANR-23-ERCB-0007.   P. Le Doussal was supported by ANR grant ANR-23-CE30-0020-01 (EDIPS).

\subsection*{Data availability statement} No datasets were generated or analysed during the current study.

\subsection*{Conflicts of interest}   The authors have no relevant financial or non-financial interests to disclose.

\printbibliography

\newpage 
\appendix
\begin{center}
    \Large\bfseries Appendix
\end{center}
\section{Properties of the kernel $\kk$}
\label{sec:operatork}
The action of the operator $\kk$ defined in \eqref{defoperator} can be interpreted as follows. Consider meromorphic functions $f(w)$ and $g(w)$ which  have no pole to the right of $\I \R$ and  converge to zero fast enough when $\Real[w]\to+\infty$. More concretely, one could assume, for instance,  that $f$ is of the form  $f(w)=\int_0^{+\infty}e^{-wt} \hat f(t)$ with $\hat f$ integrable, and likewise for $g$.
Then, if we define 
\be 
\phi(w) =  f(w)-f(w+1)+ g(-w)-g(1-w), 
\ee 
we have 
\be 
\kk \phi(w) = 2 f(1+w) + 2g(1-w). 
\ee 
There are several ways to prove this, but the simplest is to use complex analysis. Let us first focus on the integral 
\be 
\int_{\I\R } dw' \left(\psi(1+w-w') + \psi(1+w'-w )\right) \left( f(w') -f( w'+1 \right).
\ee 
When we shift the contour to the right, we pick poles of the first digamma function at points $w'=w+i$ for $i\in \mathbb Z_{\geq 1}$. Using the decay of the function $f$ as the real part increases, the integral is simply a sum of residues corresponding to singularities of $\psi(1+w-w')$.  Cauchy's residue theorem yields 
\begin{align} 
\int_{\I\R } dw' \left(\psi(1+w-w') + \psi(1+w'-w )\right) \left( f(w') -f( w'+1) \right)  &= - \sum_{i=1}^{+\infty} [f(w+i)-f(w+i+1)] \\ &= -f(w+1). 
\end{align}
Similarly, by shifting contours to the left, we find that 
\be 
\int_{\I\R } dw' \left(\psi(1+w-w') + \psi(1+w'-w )\right) \left( g(-w)-g(1-w) \right) = -g(1-w).
\ee 
\begin{remark} 
Another approach would be to use the functional equation for the digamma function.  Using $\psi(x+1)=\psi(x)+ \frac{1}{x}$, we may also show that for a function $\phi$ having no poles to the right of $\I\R$ and  decays at $+\infty$, we have that for $-1<\Real[w]<0$, 
$\kk f(w+1)-\kk f(w)=f(w+1)$. A similar equation holds for functions with no poles to the left of $\I\R$.  
\end{remark}

\section{Second cumulant in the periodic case: comparison with \cite{brunet2000probability} }
\label{app:comparisonsecondcum}

Here we show that our formula \eqref{c2per2} for the second cumulant in the periodic case
agrees with the result of \cite{brunet2000probability}. Let us recall our formula
\be 
c_2^{\rm per}(L) =  \frac{1}{2}   \partial_L \left( \int_{\I\R} \nu(dw_1) \int_{\I\R} \nu(dw_2) 
( \kk(w_1,w_2) +  \delta(w_1-w_2) ) \right) 
\ee
where we recall the definition of the delta function in \eqref{defdelta}, and
that $\nu(dw) =  \sqrt{2 \pi L} e^{\frac{L}{2} w^2} \frac{dw}{2 \I \pi}$ 
integrates to unity on the imaginary axis. Using the integral representation 
\eqref{properties} of the digamma function $\psi(z)$, and the definition of $\kk$ in \eqref{defkernel} 
we can write 
\be 
\kk(w_1,w_2)= 4 \gamma_E + 4 \int_0^{+\infty} \frac{dt}{e^t-1} (\cosh t (w_1-w_2)  - 1)
\ee 
Using that 
\be
 \int_{\I \mathbb{R}} \nu(dw_1) \int_{\I \mathbb{R}} \nu(dw_2) \cosh t (w_1-w_2) = e^{-t^2/L} 
\ee 
we obtain
\be \label{res1} 
c_2^{\rm per}(L) =  \partial_L \left( 2 \int_0^{+\infty}  \frac{dt}{e^t-1} (e^{-t^2/L}  - 1) + \frac{\sqrt{\pi L}}{2}  \right) 
\ee 
\\

On the other hand the result of Brunet and Derrida is, from setting $c=L/2$ in
formula (49) in \cite{brunet2000probability} 
\be \label{res2} 
c_2^{\rm per,BD}(L) = 
 - \partial_n^2 E|_{n=0} = \frac{1}{2 \sqrt{2 L}} \int_0^{+\infty} d\lambda \frac{\lambda^2}{\tanh( \frac{\lambda \sqrt{L}}{2 \sqrt{2}})} e^{-\lambda^2/2} 
\ee 
To compare we set $t= \lambda \sqrt{L/2}$ and we use
\bea 
\frac{2}{e^t-1} = \frac{1}{ \tanh(t/2)} - 1 
\eea 
and it is then easy to check that the two results \eqref{res1} and \eqref{res2} agree.

\section{Periodic case using the kernel method }
\label{app:comparisongeneral}

Let us now briefly indicate how to perform the Brunet-Derrida calculation using our conventions
and integral representations in terms of the kernel $K$. The Bethe equations for the periodic case, 
Eqs. (17) from \cite{brunet2000probability}, can
be obtained from \eqref{Betheab} formally setting $a=b=+\infty$ and replacing 
$e^{4 q_\alpha} \to e^{2 q_\alpha}$. The function $M(z)$ is now defined
as $M(z) = \frac{1}{n} e^{cz/2}  \sum_{\alpha\in I'_n}  \rho_\alpha e^{q_{\alpha}z}$
where $I'_n=\lbrace 1,\dots, n\rbrace$. Hence $2n$ is replaced by $n$
everywhere and one now has to solve
\be 
M(z)- M(-z) = n c \int_0^z dz' M(-z') M(z-z')
\label{eq:equationforM2}
\ee 
with $M(0)=1$, $M'(0)=n c/2$, 
together with the Bethe equation $M(z) = e^{c(z+1)} M(-z-2)$.  
These equations are equivalent to  \cite[Eq. (26-29)]{brunet2000probability}
with the correspondence \eqref{MtoB}. From  \cite[Eq. (16)]{brunet2000probability} the ground state
energy reads (this is the analog of \eqref{E0nL} above)
\be  \label{E0nLper}
E_0(n,L) = - \frac{2 n}{L^2} \left( M''(0) - \frac{1}{12} c^2 (1 + 2 n^2) \right) 
\ee 
from which the cumulants can be extracted, using 
$- E_0(n,L) = \sum_{k \geq 1} \frac{c_k^{\rm per}(L)}{k!} n^k$.
From here our calculation deviates from the one of \cite{brunet2000probability}. 
Performing exactly the same steps as in Section \eqref{subsec:expansion1},
we obtain that \eqref{eq:FBEhat2} is replaced by 
$\hat M(w) = \hat M(1-w)e^{-c + 2cw}$. The lowest order 
in the expansion in $n$ is thus given by
\be  
\hat m_0(w) = \frac{\Psi(w)}{\int_{\I\R }\frac{dw'}{2\I\pi} \Psi(w') } = 2 \sqrt{c \pi} e^{c w^2},  \quad  \quad 
\Psi(w)= \Psi^{\rm per}(w) = e^{c w^2}  = e^{\frac{L}{2} w^2}
\ee  
since we recall that $c= \gamma L/2$
and we use $\gamma=1$. We thus recover the result \eqref{nuperiodic} from the first method.
This corresponds to $m_0(z) = e^{- c z^2/4 }$ in agreement 
with $B_0(u)=1$ in (32) in \cite{brunet2000probability}. 

We can then follow the same steps as in Section \ref{subsec:systematic},
with the only change $2 n \to n$ in the equations
from \eqref{eqfrestart} to \eqref{xx}. Defining the same expansion in $n$ as in \eqref{expn},
the order zero gives, using \eqref{E0nLper}, the same result for the first
cumulant $c_1^{\rm per}(L)$ as in \eqref{c1per}, equivalent to the one obtained
in (49) in \cite{brunet2000probability}. 

To next order one must solve \eqref{eqf10} without the factor of $2$. The solution if thus 
\be 
f_1(w_1) = \int_{\I\R+\delta} \nu(dw_2) K(w_1,w_2) - C'_1
\label{eq:deff1per}
\ee 
where $K$ is defined in \eqref{eq:defK} and
$C'_1$ is a constant such that $\int_{\I\R} \nu(dw_1) f_1(w_1) =0$ (where we
recall that $\nu(dw_1)$ is given in \eqref{nuperiodic}).
Hence we obtain for the function $m_1(z)$
\be
m_1(z) 
 = 4 \pi c \int_{\I\R} \frac{dw_1}{2\I \pi} e^{c w_1^2} e^{c z w_1} \int_{\I\R+\delta} \frac{dw_2}{2\I \pi} e^{c w_2^2} 
K(w_1,w_2) - C'_1 m_0(z) 
\ee 
Let us use the representation for ${\rm Re} (w_1+w_2)>0$.
\be 
K(w_1,w_2) = \int_0^{+\infty} dt \left[ e^{- t(w_1+w_2)} 
 + \frac{2}{e^t-1} (\cosh t (w_1-w_2)  - 1) \right] + 2 \gamma_E 
\ee 
We obtain
\be 
m_1(z) = 
e^{- \frac{c z^2}{4} } \bigg( \int_0^{+\infty} dt \left[ e^{\frac{1}{2} t z - \frac{t^2}{2 c} } 
+ \frac{2}{e^t-1} 
\left[ \cosh\left( \frac{t z}{2} \right)
   e^{ - \frac{t^2}{2 c} } - 1 \right] \right] 
   + 2 \gamma_E - C'_1 \bigg) 
\ee 
Taking into account that $m_1(0)=0$ determines $C'_1$ and we obtain
\be 
m_1(z) = 
e^{- \frac{c z^2}{4} }  \int_0^{+\infty} dt \left[ e^{\frac{1}{2} t z }-1 
+ \frac{2}{e^t-1} \left( \cosh\left( \frac{t z}{2} \right) - 1\right) 
     \right] e^{ - \frac{t^2}{2 c} }
\ee 
We now compare with the result of \cite{brunet2000probability}. From eq. (42) there
we can read $b_1(u)$ which using \eqref{MtoB} implies the result for $m_1(z)$
\be 
m_1(z) = m_0(z) \sqrt{c} \int_0^{+\infty} d\lambda e^{- \lambda^2/2} \frac{\cosh(\frac{\lambda (1+z) \sqrt{c}}{2}) 
- \cosh(\frac{\lambda  \sqrt{c}}{2})}{ \sinh(\frac{\lambda  \sqrt{c}}{2}) } \quad , \quad m_0(z)= e^{\frac{-c z^2}{4}} 
\ee 
Setting $\lambda=t/\sqrt{c}$ we see that the two expressions are identical, thanks to the trigonometric
identity
\be
e^{\frac{1}{2} t z }-1
+ \frac{2}{e^t-1} \left( \cosh\left( \frac{t z}{2} \right) - 1\right) = \frac{ \cosh\left( \frac{1+z}{2} t\right) - \cosh \left(\frac{t}{2}\right) }{ \sinh (\frac{t}{2}) }
\ee 
\\
Hence our calculation using the kernel $K$ gives the same function $M(z)$ as the one of 
of \cite{brunet2000probability} at least to the order $O(n^2)$ in the ground state energy.
It shows, among other things, that the way we have fixed some 
arbitrariness in the intermediate steps of the calculation is equivalent to the
way it was done in \cite{brunet2000probability} (as we stressed
above this is presumably related to choosing the ground state
rather than some excited state).
Note that in \cite{brunet2000probability} a recursive construction was
given for the higher orders, i.e. $O(n^q)$ with $q \geq 3$. We have not
attempted to pursue the identification to the higher orders.

\section{Some explicit formula in the  large $L$ limit} 
\label{app:formula}

Let us recall the definitions \eqref{defphipsi}
\be \label{defphipsi3} 
\phi_k=\int_{\R} \frac{dy}{2\pi}  \phi(y)^k, \;\;\; \psi_k=\int_{\R} \frac{dy}{2\pi}  y^2 \phi(y)^k,
\quad  \quad \phi(y) = \frac{4 y^2 e^{-y^2} }{(\widetilde u^2+y^2)(\widetilde v^2+y^2)}. 
\ee 
Here we first provide a systematic method to evaluate these coefficients for general $\tilde u, \tilde v > 0$ in terms of error functions. Let  $a=\tilde u^2$ and $b=\tilde v^2$ and define 
\be 
f(a) := \int_{\mathbb R} \frac{dy}{2 \pi} \frac{e^{-y^2}}{y^2 + a} = \frac{e^a  \text{erfc}\left(\sqrt{a}\right)}{2 \sqrt{a}}. 
\ee 
Then we may write
\be 
\int_{\R} \frac{dy}{2 \pi} \frac{e^{-y^2}}{(y^2 + a)(y^2+b)} = \frac{f(a)-f(b)}{b-a}.
\ee 
Next, one has, for $t>0$ and positive integer $k$, 
\begin{equation}
\int_{\R} \frac{dy}{2 \pi} \frac{e^{-t y^2}}{(y^2 + a)^k(y^2+b)^k} 
= \frac{1}{\Gamma(k)^2} t^{\frac{3}{2} } \partial_{a}^{k-1} \partial_{b}^{k-1} 
\int_{\R} \frac{dy}{2 \pi} \frac{e^{-y^2}}{(y^2 + a t)(y^2+b t)} 
\end{equation}  
Hence one can write 
\begin{align}  \label{phigen} 
 \phi_k  
&= \frac{4^k}{\Gamma(k)^2} (- \partial_t)^k \partial_{a}^{k-1} \partial_{b}^{k-1} \left( t^{\frac{1}{2} } 
\frac{f(a t)-f(b t)}{b-a}  \right) \bigg\vert_{t=k}\\
\label{psigen} 
 \psi_k 
&= \frac{4^k}{\Gamma(k)^2} ( \partial_t)^{k+1} \partial_{a}^{k-1} \partial_{b}^{k-1} \left( t^{\frac{1}{2} } 
\frac{f(a t)-f(b t)}{b-a}  \right) \bigg\vert_{t=k}
\end{align}
One can check that these coefficients obey the relation \eqref{relationphipsi}.
There are special cases where these formula become more explicit. Using the integral
\begin{align*}
 I_{k,p,m}(a) &= \int_{ \mathbb{R}} \frac{dy}{ 2 \pi} \frac{y^{2 p}}{(a + y^2)^{m}}  e^{-k y^2} 
\\ 
&= \frac{1}{\Gamma(m)} \int_0^{+\infty} dt t^{m-1} 
\int_{ \mathbb{R}} \frac{dy}{ 2 \pi}  y^{2 p}  e^{- y^2 ( k + t) - t a }  \\
&= \frac{\Gamma(p+1/2)}{2 \pi \Gamma(m)}  \int_0^{+\infty} dt \, t^{m-1} \frac{e^{-t a}}{(k+t)^{p + \frac{1}{2}}}\\
&= \frac{\Gamma(p+1/2)}{2 \pi } a^{p+\frac{1}{2} - m} U(p+\frac{1}{2},p+ \frac{3}{2} - m,k a) 
\end{align*}
where $U$ is the 
confluent hypergeometric function. We
can now consider several cases. 

\subsection{The case $\tilde u=\tilde v$} In that case one finds
\bea  
&& \phi_k = 4^k I_{k,k,2k}(\tilde u^2) = 
4^k  
 \frac{\Gamma(k+ \frac{1}{2} ) \tilde u^{1- 2 k} }{2 \pi} U\left(k+\frac{1}{2},\frac{
   3}{2}-k,k \tilde u^2\right)  \\
&& \psi_k = 4^k I_{k,k+1,2k}(\tilde u^2) = 4^k 
\frac{\Gamma(k+\frac{3}{2}) \tilde u^{3- 2 k} }{2 \pi}  U\left(k+\frac{3}{2},\frac{5}{2}-k,k u^2\right)
\eea

For instance for $k=1$ it gives
\bea
&& \phi_1= \frac{e^{{\tilde u}^2} \left(2
   {\tilde u}^2+1\right)
   \text{erfc}({\tilde u})}{{\tilde u}}-\frac{2}{\sqrt{\pi }} \\
&& \psi_1= \frac{2
   \left({\tilde u}^2+1\right)}{\sqrt{\pi
   }}-e^{{\tilde u}^2} {\tilde u} \left(2
   {\tilde u}^2+3\right) \text{erfc}({\tilde u})   
\eea 
which, one can check, agrees with \eqref{phigen}, \eqref{psigen}.

These formula then allow to compute the scaled cumulants $\tilde c_k$, e.g.
the first two are given by
\be
\widetilde c_1 =  -\frac{\psi _1}{2 \phi_1 },  \quad  \quad 
\widetilde c_2 = 
\frac{1}{\phi_1^3}  \left(\psi _1 \phi _2-\psi_2 \phi_1\right) 
\ee 
Let us give here their asymptotic behavior, at small $\tilde u$
\bea  \label{equalsmallu}
&& \widetilde c_1 = -\frac{{\tilde u}}{\sqrt{\pi
   }}+\left(\frac{3}{2}-\frac{4}
   {\pi }\right)
   {\tilde u}^2+O\left({\tilde u}^3\right), \\
&& \widetilde c_2 = \frac{1}{\sqrt{\pi
   }}+\left(\frac{12}{\pi
   }-4\right)
   {\tilde u}+O\left({\tilde u}^2\right), 
\eea  
and at large $\tilde u$
\bea \label{equallargeu}
&& \widetilde c_1 = -\frac{3}{4}+\frac{3}{2
   {\tilde u}^2}+O\left(\frac{1}{{\tilde u}^3}
   \right),  \\
&& \widetilde c_2 = \frac{3 \sqrt{\frac{\pi
   }{2}}}{8}-\frac{3
   \sqrt{\frac{\pi }{2}}}{8
   {\tilde u}^2}+O\left(\frac{1}{{\tilde u}^3}
   \right).
\eea 
While $\widetilde c_1 \leq 0$ is decreasing with $\tilde u$, one finds
that $\widetilde c_2$ first decreases from its finite value at $\tilde u=0$
and then increases, exhibiting a minimum value $0.446153$
at $\tilde u= 2.19956$. Note that the limit $\tilde u \to +\infty$
agrees with the result \eqref{cumulantsmaximal}. Note that although $\phi_k$ and $\psi_k$ diverge as $\tilde u \to 0$,
the cumulants $\widetilde c_k$ have a finite limit at the
point $(\tilde u, \tilde v)=(0,0)$. For instance, 
it is easy to check that the $\tilde u=0$ limit of $\widetilde c_2$ 
in \eqref{equalsmallu}, agrees with
the large $L$ limit from 
\eqref{eq:c200}, using \eqref{cklargetime2}.

\bigskip 

Let us give also the asymptotics of the third cumulant $\tilde c_3$
\be \label{c3new} \widetilde c_3 = 
\begin{cases} \frac{4 (\pi -3)}{\pi }-\frac{48
   \left(5-3 \pi +\sqrt{2} \pi
   \right) \tilde u}{\pi
   ^{3/2}}+O\left(\tilde u^2\right) \\
  \left(\frac{80}{27
   \sqrt{3}}-\frac{27}{16}\right
   ) \pi
   +\frac{\left(\frac{160}{9
   \sqrt{3}}-\frac{675}{64}\right) \pi
   }{\tilde u^4}+O\left(\frac{1}{
   \tilde u^5}\right) \end{cases} 
\ee 
One finds that $\widetilde c_3$ is a positive decreasing function of $\tilde u$.
Note that the limit $\tilde u \to +\infty$
agrees with the result \eqref{cumulantsmaximal}. For $\tilde u \to 0$,
one can push the exercise to higher orders and one finds
\be \label{cum45}
\widetilde c_4 = \frac{48
   \left(5+\left(\sqrt{2}-3\right) \pi \right)}{\pi ^{3/2}} + O(\tilde u),  \quad \quad 
   \widetilde c_5 = -\frac{320 \left(21+2 \left(4
   \sqrt{2}-9\right) \pi
   \right)}{\pi ^2} + O(\tilde u).
\ee 
This gives some higher cumulants for the point $(\tilde u,\tilde v)=0$.


\subsection{The case $\tilde v=0$, $\tilde u>0$, or $\tilde u=0$, $\tilde v>0$} For $\tilde v=0$, the 
coefficients are well defined for $\tilde u>0$ and read
\bea  \label{phipsiv0}
&& \phi_k = 4^k I_{k,0,k}(\tilde u^2) = 4^k  \frac{\tilde u^{1-2 k}}{2 \sqrt{\pi}} U\left(\frac{1}{2},\frac{3}{2}-k,k \tilde u^2\right) \\
&& \psi_k = 4^k I_{k,1,k}(\tilde u^2) = 4^k \frac{\tilde u^{3- 2 k} }{4 \sqrt{\pi }} U\left(\frac{3}{2},\frac{5}{2
   }-k,k \tilde u^2\right) 
\eea  
with the same result for $\tilde u=0$ as a function of $\tilde v$.
One finds that the first cumulant $\widetilde c_1$ is negative and a decreasing function of $\tilde u$
\bea 
\widetilde c_1 =  -\frac{E_{\frac{3}{2}}\left({\tilde u}^2\right)}{4
   E_{\frac{1}{2}}\left({\tilde u}^2\right)} = -\frac{{\tilde u}^2 \Gamma
   \left(-\frac{1}{2},{\tilde u}^2\right)
   }{4 \sqrt{\pi }
   \text{erfc}\left({\tilde u}\right)} = \begin{cases} -\frac{1}{4}+\frac{1}{4
   {\tilde u}^2}+O\left(\frac{1}{{\tilde u}^3} \right) \\ 
   -\frac{{\tilde u}}{2\sqrt{\pi
   }}+\frac{(\pi -2) {\tilde u}^2}{2 \pi
   }+O\left({\tilde u}^3\right) \end{cases}
\eea 
where $E_a(z)=\int_1^{+\infty} e^{-z t}/t^a$ is the exponential integral function
and $\Gamma(a,z)$ the incomplete Gamma function. The second cumulant has again a minimum at ${\tilde u}=0.423115$
of value $0.548349$, with asymptotics
\bea 
\widetilde c_2 = \begin{cases} 
\frac{1}{\sqrt{\pi }}-\frac{2
   (\pi -3) {\tilde u}}{\pi
   }+\frac{\left(24-13 \pi +4
   \sqrt{2} \pi \right) {\tilde u}^2}{\pi
   ^{3/2}}+O\left({\tilde u}^3\right)
\\
\frac{\sqrt{\pi}}{2 \sqrt{2}}-\frac{\sqrt{\frac{\pi }{2}}}{4
   {\tilde u}^2}+O\left(\frac{1}{{\tilde u}^3}
   \right)
 \end{cases}
\eea 
Note that the limit $\tilde u\to 0$ coincides with the value found above
in the case $\tilde v=\tilde u \to 0$. More generally, 
pushing the expansion to obtain higher cumulants (up to order $5$ included)
one recovers exactly \eqref{c3new} and \eqref{cum45}, which we
can identify with the result for the point $(\tilde u,\tilde v)=0$.

\bigskip 

In the limit $\tilde u \to + \infty$ (with $\tilde v=0$) one finds from \eqref{phipsiv0}
\be 
\phi_k \simeq \frac{2^{2 k-1} {\tilde u}^{-2
   k}}{\sqrt{\pi } \sqrt{k}},  \quad  \quad \psi_k = \frac{4^{k-1} {\tilde u}^{-2
   k}}{\sqrt{\pi } k^{3/2}}
\ee 
which, from \eqref{parametricrescaled} leads to the following 
parametric representation of the cumulant generating function $E(\widetilde s)$
\be \label{parametricrescaled2}
\widetilde s = \frac{1}{4 \sqrt{\pi} } {\rm Li}_{3/2}(\hat \zeta),  \quad \quad \widetilde E(\widetilde s) = -\frac{1}{16 \sqrt{\pi} } {\rm Li}_{5/2}(\hat \zeta)
\ee 
where $\hat \zeta=4 \tilde \zeta/u^2$ must be eliminated.
This can be compared with the result for the 
ASEP in \cite[Eq. (IV.70-71)]{LazarescuThesis} on the HD-MC transition line.
We find again that it agrees, using \eqref{corresp}, except
that the non-universal shift $\tilde s/2$ in the second line of \eqref{corresp} is now absent.

\section{First cumulant as a functional of Brownian motions} 
\label{sec:appendixfunctionals}
Consider the SHE \eqref{eq:SHE} on some domain $\mathbb X$ (either the interval  $[0,L]$ with periodic or open boundary conditions, or the full-line $\R$). Let 
\be
h_{\varphi}(t):= \log\left( \int_{\mathbb X}\mathrm dx \varphi(x) Z(x,t) \right).
\ee 
As $t$ goes to infinity, we expect that $h_{\varphi}(t)\sim c_1 t$ where $c_1$ does not depend on $\varphi$. Using \eqref{eq:SHE} and Ito equation, 
\be 
\frac{d}{dt} \EE\left[ h_{\varphi}(t) \right] = \EE\left[ \frac{\int_{\mathbb X} dx \varphi(x) \partial_{x}^2 Z(x,t) }{\int_{\mathbb X} dx \varphi(x)  Z(x,t)} \right] - \frac{1}{2 } \EE\left[ \frac{\int_{\mathbb X} dx \varphi(x)^2 Z(x,t)^2 }{\left(\int_{\mathbb X} dx \varphi(x)  Z(x,t)\right)^2} \right]. 
\ee
Assume that $\varphi$ is smooth and satisfies the following boundary conditions:
\begin{itemize}
    \item $\varphi(0)=\varphi(L)$ if $\mathbb X=[0,L]$ with periodic boundary conditions $Z(0,t)=Z(L,t)$ as in Section \ref{sec:periodic};  
    \item $\varphi'(0)= (u-1/2)\varphi(0)$ and $\varphi'(L)= -(v-1/2)\varphi(L)$ if  $\mathbb X=[0,L]$ with  open boundary conditions as in \eqref{eq:SHEboundaryconditions} (for more details about the boundary condition imposed on the test function, we refer to \cite[Section 5]{corwin2018open}); 
\end{itemize}
Integrations by parts yield 
\be 
\frac{d}{dt} \EE\left[ h_{\varphi}(t) \right] = \EE\left[ \frac{\int_{\mathbb X} dx \varphi''(x) Z(x,t) }{\int_{\mathbb X} dx \varphi(x)  Z(x,t)} \right] - \frac{1}{2 } \EE\left[ \frac{\int_{\mathbb X} dx \varphi(x)^2 Z(x,t)^2 }{\left(\int_{\mathbb X} dx \varphi(x)  Z(x,t)\right)^2} \right]. 
\ee 
Dividing numerators and denominator by $Z(0,t)$ and using the fact that $Z(x,t)/Z(0,t)$ converges in distribution to $e^{H(x)}$ at large time, we obtain that as $t$ goes to infinity, 
\be 
c_1 = \lim_{t\to\infty} \frac{d}{dt} \EE\left[ h_{\varphi}(t) \right] = \EE\left[ \frac{\int_{\mathbb X} dx \varphi''(x) e^{H(x)} }{\int_{\mathbb X} dx \varphi(x)  e^{H(x)}} \right] - \frac{1}{2 } \EE\left[ \frac{\int_{\mathbb X} dx \varphi(x)^2 e^{2H(x)} }{\left(\int_{\mathbb X} dx \varphi(x)  e^{H(x)}\right)^2} \right].
\label{eq:functionalforc1}
\ee 
In the case where $\mathbb X=\R$, the limiting stationary measure is $H(x) = B(x)$, a two-sided Brownian motion without drift. In the case of the interval with periodic boundary conditions, the unique stationary measure $H$ is a Brownain bridge, and  with open boundary conditions,  $H$ is described in Section \ref{sec:stationarymeasures}. 
\begin{remark}
    It is remarkable that the formula \eqref{eq:functionalforc1} should be independent of the test function $\varphi$.  This is a strong constraint, and  it would be interesting to understand how much this constraint characterizes the stationary measure. 
\end{remark}
\subsection*{Special cases}
In the case of periodic boundary conditions, we may take $\varphi\equiv 1$, and we find 
\be
c_1^{\rm per}(L) =  - \frac{1}{2 } \EE\left[ \frac{\int_0^L dx  e^{2B(x)} }{\left(\int_0^L dx   e^{B(x)}\right)^2} \right] = -\frac{L}{2}  \EE\left[ \frac{1}{\left(\int_0^L dx  e^{B(x)}\right)^2} \right]
\ee
where $B$ is a Brownian bridge on $[0,L]$ starting from $B(0)=0$. We recover the formula obtained in \cite[Eq. (4.1)]{dunlap2023fluctuation}, which can be matched with \eqref{c1per} (see \cite[Prop. 4.1]{dunlap2023fluctuation}).

In the case of open boundary conditions with $u=v=1/2$, we may also take $\varphi\equiv 1$, and we find 
\be 
c_1(1/2,1/2,L) =  - \frac{1}{2 } \EE\left[ \frac{\int_0^L dx  e^{2H(x)} }{\left(\int_0^L dx   e^{H(x)}\right)^2} \right]
\ee 
where $H$ is the stationary measure described in Section \ref{sec:stationarymeasures}. When 
$u+v=1$, one can choose $\varphi(x) =e^{x(u-1/2)}$ and using \eqref{eq:stationarymeasure2}, we find 
\be 
c_1(u,1-u,L) = \frac{1}{\mathcal Z_{u,1-u,L}(L)}\mathbb E\left[ \frac{e^{-v(B_1(L)-B_2(L))}\int_0^L e^{2(B_1(x) + (u-1/2)x)}dx}{\left( \int_0^L e^{B_1(x)+(u-1/2)x} dx \right)^2 \int_0^L e^{-(B_1(x)-B_2(x))} dx }\right]
\ee 
where $B_1$ and $B_2$ are independent standard Brownian motions and 
\be 
\mathcal Z_{u,1-u,L}(L) = \frac{1}{2} \int_{\I\R} \frac{dw}{2\I\pi} \Psi(w) = \mathbb E\left[ \frac{e^{-v(B_1(L)-B_2(L))}}{\int_0^L e^{-(B_1(x)-B_2(x))} dx }\right].
\ee

\end{document}